\documentclass[journal]{IEEEtran}  
\usepackage{hyperref}
\usepackage{url}
\usepackage{amsmath,amssymb}
\usepackage{amsthm}
\allowdisplaybreaks
\usepackage{mathtools}
\usepackage{bbm}
\usepackage{multirow}

\usepackage{graphicx}
\usepackage{bm}
\usepackage{textcomp}
\usepackage{algorithm}
\usepackage[noend]{algpseudocode}
\usepackage{bbm}
\usepackage{cite}
\usepackage{xprintlen}
\usepackage{cases}
\usepackage{tabularx,ragged2e}
\usepackage{booktabs}
\usepackage{nccmath}
\usepackage{caption}
\usepackage{float}
\usepackage{subcaption}
\usepackage{mathtools}

\renewcommand{\vec}[1]{\bm{#1}}

\newcommand{\true}{\top}
\newcommand{\false}{\bot}

\newtheorem{problem}{\textbf{Problem}}
\newtheorem{proposition}{\bf{Proposition}}
\newtheorem{definition}{\bf{Definition}}
\newtheorem{lemma}{\bf{Lemma}}
\newtheorem{theorem}{\bf{Theorem}}

\newtheorem{example}{\textbf{Example}}
\newtheorem{brule}{\textbf{Rule}}

\begin{document}
\title{Sampling-Based Planning Under STL Specifications: A Forward Invariance Approach}
% \author{%%%% author names
%     \IEEEauthorblockN{ Gregorio Marchesini},% first author
%     \IEEEauthorblockN{ Siyuan Liu},
%     \IEEEauthorblockN{ Dimos V. Dimarogonas}
%     \newline\indent%%%% author affiliations
%     \IEEEauthorblockA{\textit{KTH Royal Institute of Technology, Stockholm, Sweden}}\newline\indent% first affiliation
%     % duplicate the line above as many times as needed to list all affiliations
%     %%%% corresponding author contact details
% }

\author{Gregorio Marchesini, ~\IEEEmembership{Student Member,~IEEE},  Siyuan Liu, ~\IEEEmembership{Member,~IEEE}, \\
Lars Lindemann, ~\IEEEmembership{Member,~IEEE}, and Dimos V. Dimarogonas, ~\IEEEmembership{Fellow,~IEEE}% <-this % stops a space
\thanks{This work was supported in part by the Horizon Europe EIC project SymAware (101070802), 
the ERC LEAFHOUND Project, the Swedish Research Council (VR), Digital Futures, and the Knut and Alice Wallenberg (KAW) Foundation.}% <-this % stops a space
\thanks{Gregorio Marchesini, Siyuan Liu, and Dimos V. Dimarogonas are with the Division
of Decision and Control Systems, KTH Royal Institute of Technology, Stockholm, Sweden.
	E-mail: {\tt\small \{gremar,siyliu,dimos\}@kth.se}.  Lars Lindemann is with the Thomas Lord Department of Computer Science, University of Southern California, Los Angeles, CA, USA.
	E-mail: {\tt\small \{llindema\}@usc.ed}.   
}
% \thanks{ 
% }
}

\maketitle
\begin{abstract}
We propose a variant of the Rapidly Exploring Random Tree Star (RRT$^{\star}$) algorithm to synthesize trajectories satisfying a given spatio-temporal specification expressed in a fragment of Signal Temporal Logic (STL) for linear systems. Previous approaches for planning trajectories under STL specifications using sampling-based methods leverage either mixed-integer or non-smooth optimization techniques, with poor scalability in the horizon and complexity of the task. We adopt instead a control-theoretic perspective on the problem, based on the notion of set forward invariance. Specifically, from a given STL task defined over polyhedral predicates, we develop a novel algorithmic framework by which the task is efficiently encoded into a time-varying set via linear programming, such that trajectories evolving within the set also satisfy the task. Forward invariance properties of the resulting set with respect to the system dynamics and input limitations are then proved via non-smooth analysis. We then present a modified RRT$^{\star}$ algorithm to synthesize asymptotically optimal and dynamically feasible trajectories satisfying a given STL specification, by sampling a tree of trajectories within the previously constructed time-varying set. We showcase two use cases of our approach involving an autonomous inspection of the International Space Station and room-servicing task requiring timed revisit of a charging station.
\end{abstract}
\section{Introduction}\label{sec:introduction}
\begin{figure}[t]
    \centering
\includegraphics[width=\linewidth]{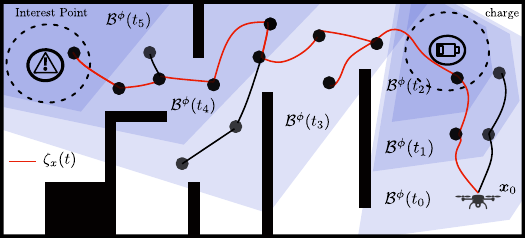}
    \caption{A trajectory $\zeta_{x}$ (red) that satisfies an STL task $\phi$ is obtained by expanding a tree of trajectories, using an RRT, from a given initial state $\vec{x}_0$. The tree of trajectories is expanded within the time-varying set $\mathcal{B}^{\phi}(t)$ such trajectory evolving into $\mathcal{B}^{\phi}(t)$ also satisfy the STL task $\phi$.}
    \label{fig:RRT expansion example}
\end{figure}

The application of sampling-based planners to plan kino-dynamically feasible trajectories in complex environments for systems subject to spatio-temporal constraints has been an active area of research during the past decades \cite{plaku2016motion,branicky2006sampling}. Particularly, with the aim of deploying autonomous systems with verifiable performance guarantees, a growing body of literature has been devoted toward designing planning algorithms to synthesize trajectories satisfying tasks expressed as Linear Temporal Logic (LTL) formulas \cite{kress2009temporal,saha2014automated,MurrayLTL,aminof2019planning,plaku2012planning} and, more recently, Signal Temporal Logic (STL) formulas \cite{kurtz2022mixed,raman2015reactive,sun2022multi,sewlia2023cooperative}.

The focus of this work is on trajectory planning for linear systems under STL specification, leveraging Rapidly Exploring Random Trees (RRT) \cite{lavalle2001randomized} and, in particular, its asymptotically optimal variant RRT$^\star$ \cite{karaman2011sampling}, with applications to real-time robot motion planning in complex environments. 
\subsubsection*{Contributions}
In this paper, we propose to adopt a novel approach to real-time sampling-based planning of trajectories for linear systems subject to STL constraints, leveraging the notion of forward invariance of time-varying sets. Namely, we formalize a simple, yet effective, approach to cast a STL specification, expressed over linear predicate functions, into a time-varying set, expressed as a time-varying polyhedron, which we design via linear programming. We show via a non-smooth analysis that the resulting set is forward invariant for the system dynamics, i.e., there exists a controller that maintains the system within the set under input constraints, given that the system starts within the set.  While the problem of casting STL specifications into time-varying sets was previously considered in \cite{lindemann2018control, charitidou2021barrier,yu2024continuous,Jiang,reach1}, our contribution differs from these previous works in the following terms. While \cite{yu2024continuous,Jiang,reach1} consider the entire STL fragment of specifications, their approach considers complex reachable set computations, which often can not be achieved in real-time. In this regard, we trade off a loss in expressivity, as we consider a subfragment of STL, for a reduction in computational complexity, for which we don't need to undertake complex reachable sets computations. On the other hand, compared to \cite{lindemann2018control, charitidou2021barrier}, we here consider a more expressive STL fragment that includes nested temporal operators, thus allowing for a richer set of specifications, e.g., recurring tasks. 

Specifically, we encode an STL task $\phi$ into a time-varying polyhedron $\mathcal{B}^{\phi}(t)$ (see. Fig \ref{fig:RRT expansion example}), we provide a real-time implementation of RRT$^\star$ to synthesize dynamically feasible trajectories satisfying the task $\phi$. Namely, if we  let $\zeta_x(t)$ represent a trajectory of a given linear system, we then propose to expand a tree of sampled trajectory that evolves within the time-varying set $\mathcal{B}^{\phi}(t)$ from which we obtain a minimum cost trajectory satisfying $\phi$.
\subsubsection*{Related work}
The authors in \cite{barbosa2019guiding,vasile2017sampling,linard2023real,karlsson2020sampling,sewlia2023cooperative} first explored the application of RRT$^\star$ to synthesize a feasible trajectory satisfying a given STL specification. Namely,  \cite{barbosa2019guiding,linard2023real,karlsson2020sampling} exploit the notion of STL robust semantics \cite{fainekos2009robustness}—a metric functional defining the distance of a trajectory from the set of trajectories satisfying a given task $\phi$—to guide the synthesis of satisfying trajectories using RRT$^\star$. However, the dynamics of the system, as well as input limitations, are not directly considered therein, which are instead considered in the present work. A different approach is considered in \cite{sewlia2023cooperative}, where an RRT algorithm is defined for a multi-agent system case. A recursive set of rules is defined to guide the expansion of the tree of sampled trajectories, avoiding the introduction of robust semantics, but the dynamics of the system are again not considered, whereas we consider controllable linear systems. On the other hand, \cite{vasile2017sampling} proposes an RRT$^{\star}$ algorithm that handles nonlinear dynamics with input limits and STL constraints via robust semantics, which are non-smooth and require recursive evaluations. We instead avoid recursive evaluation by casting STL tasks as time-varying sets offline. More in line with our approach are the works in \cite{ahmad2022adaptive,yang2019sampling,planningwithSTL} where Control Barrier Function (CBF) constraints are introduced within a sampling-based planner, but with the goal of synthesizing collision-free trajectories, rather than enforcing spatio-temporal constraints on them. Particularly in \cite{planningwithSTL}, Signal Temporal Logic constraints are considered by leveraging a Mixed-Integer Linear Programming (MILP) formulation \cite{raman2015reactive}, which is known to scale poorly over the state dimension and duration of the task.
\subsubsection*{Organization} Sections \ref{sec:preliminaries} and \ref{sec:problem} introduce preliminaries and problem formulation. In Sec. \ref{sec:stl tasks} and \ref{sec:forward invariance verification} our first contribution is provided, developing an algorithmic approach to construct a time-varying polyhedral set from a given STL task $\phi$, with guaranteed forward invariance properties, such that trajectories evolving in this set also satisfy $\phi$. In Sec. \ref{sec:rrt} an implementation of RRT$^{\star}$ leveraging such a time-varying set is shown, while simulations showcasing our algorithms are given in Sec. \ref{sec:sim}. Conclusions and future work are provided in Sec. \ref{sec:conclusion}.
\subsubsection*{Notation} Bold letters denote column vectors and capital letters indicate matrices and sets. The sets $\mathbb{N}$, $\mathbb{R}$ and $\mathbb{R}_{\geq 0}$ represent natural, real, and non-negative real numbers, respectively. The symbols $\top$ and $\bot$ denote logical \textit{True} and \textit{False}, respectively. Given two metric spaces $Y$ and $Z$,  the set $\mathcal{F}(Y,Z)$ is the set of functions from $Y$ to $Z$. For a set $\mathcal{S}$, $|\mathcal{S}|$ and $2^{\mathcal{S}}$ represent cardinality and power set (set of subsets), respectively. For two sets $A$ and $B$, the Cartesian product is $A\times B = \{(x,y)\mid x\in A\;, y\in B\}$. For a given interval $[a,b]\subset \mathbb{R}$ and $c \in \mathbb{R}$, we let $c \oplus[a,b] =  [c+a,c+b]$. Standard 2-norm and $\infty$-norm are denoted as $\|\vec{x}\|$ and $\|\vec{x}\|_{\infty}$, respectively,  while $\|\vec{x}\|_Q:= \sqrt{\vec{x}^T Q\vec{x}}$ for a positive definite matrix $Q$. For a real-valued function $g: \mathbb{R}^n \times \mathbb{R} \rightarrow \mathbb{R} $, we denote the standard gradient as $\nabla g(\vec{x},t) = [\frac{d}{d\vec{x}}g(\vec{x},t)\; \frac{d}{dt}g(\vec{x},t)]^T \in \mathbb{R}^{n+1}$. An extended class-$\mathcal{K}$ function $\kappa : \mathbb{R} \rightarrow \mathbb{R}$ is a strictly monotonically increasing function and such that $\kappa(0) = 0$. Vector inequalities $A\vec{x} \leq \vec{b}$ with $\vec{x} \in \mathbb{R}^n$, $A\in \mathbb{R}^{p\times n}$ and $\vec{b} \in \mathbb{R}^p$ are interpreted element-wise. A set $\mathcal{P} = \{\vec{x}\mid Ax\leq \vec{b} \}$ is a \textit{polyhedron} and it is a polytope if it is bounded. Ordered sequences are defined as $(s_i)_{i=1}^{n}$ with $s_{i+1} \geq s_i$, while (unordered) sets are defined as $\{s_i\}_{i=1}^{n}$. For an integer $n\in \mathbb{N}$, let the index set $[[n]] := \{1 ,\ldots n\}$. For a set $\{\vec{x}_k\} \subset \mathbb{R}^{n}$, the \textit{convex hull} is $\text{co}(\{\vec{x}_k\}) := \{\sum_{k}\lambda_k\vec{x}_k \mid \lambda_k \geq 0, \; \sum_{k}\lambda_k=1 \}$. The symbols $0_{n\times m}$ and $\bm{1}_n$ denote a zero matrix and a vector of ones.
\section{Preliminaries}\label{sec:preliminaries}
Consider the continuous-time linear dynamics 
\begin{equation}\label{eq:single agent dynamics}
\dot{\vec{x}} = f(\vec{x},\vec{u}) =A\vec{x} + B\vec{u} + \vec{p},
\end{equation}
where $\vec{x}\in \mathbb{R}^{n} \subseteq \mathbb{X}$ and $\vec{u} \in \mathbb{R}^{m} \subseteq \mathbb{U}$ are state and input vectors, respectively, while $\vec{p}\in \mathbb{R}^{n}$ is a constant vector (e.g. a constant off-set). Let $(A,B)$ be a controllable pair, while $\mathbb{X}$ and $\mathbb{U}$ are bounded convex polytopes of the form
\begin{equation}
\begin{aligned}
    \mathbb{X} := \{ \vec{x}\in \mathbb{R}^{n} \mid D_x \vec{x}\leq \vec{c}_x \},\\
    \mathbb{U} := \{ \vec{u}\in \mathbb{R}^{m} \mid D_u \vec{u}\leq \vec{c}_u \},
\end{aligned}
\end{equation}
where $D_{x}\in \mathbb{R}^{p_x\times n}$, $D_{u}\in \mathbb{R}^{p_u\times m}$, $\vec{c}_u \in \mathbb{R}^{p_{u}}$ and $\vec{c}_{x} \in \mathbb{R}^{p_{x}}$, with $p_x\geq n+1$ and $p_u\geq m+1$. Since $\mathbb{X}$ is a closed and bounded polytope, there exists a finite set of vertices $V_{\mathbb{X}} = \{\vec{\eta}_k\} \subset \mathbb{R}^{n}$, with $n^v_{\mathbb{X}} = |V_{\mathbb{X}}|$, such that $\mathbb{X} = co(\{\vec{\eta}_k\})$. Thus each $\vec{x}\in \mathbb{X}$ is a convex combination $\vec{x}= \sum_{k=1}^{n_{\mathbb{X}}^v} \lambda_k \vec{\eta}_k$, where $\lambda_k\geq 0$ and $\sum_{k=1}^{n_{\mathbb{X}}^v} \lambda_k =1$ \cite[Prop. 2.2]{ziegler2012lectures}. We later leverage the discrete vertex representation of $\mathbb{X}$ to infer invariance properties of time-varying sets defined over $\mathbb{X}$.

We hereafter consider Carathéodory solutions to \eqref{eq:single agent dynamics}, i.e., absolutely continuous functions (or trajectories) $\zeta_x : [t_0,t_1] \rightarrow \mathbb{X}$ that satisfy 
\begin{subequations}
\begin{align}
\zeta_x(t) &= \int_{t_0}^{t} f(\zeta_x(t), \zeta_u(t)) dt + \vec{x}_0, \\
\dot{\zeta}_x(t) &= f(\zeta_x(t), \zeta_u(t)) \quad a.e. \;  t\in [t_0,t_1],\label{eq:dynamics constraint}
\end{align}
\end{subequations}
with initial state $\vec{x}_0 \in \mathbb{X}$ and measurable control input signal $\zeta_u: [t_0,t_1] \rightarrow \mathbb{U}$ \cite[Prop. S1]{cortes2008discontinuous}. Caratheodory solutions are commonly considered in the analysis of non-smooth dynamical systems e.g., when  Lipschitz continuous inputs are not sufficiently rich to enforce stability/invariance properties as in the settings of our work.

\subsection{Signal Temporal Logic}
Signal Temporal Logic (STL) is a predicate logic suitable to define spatial and temporal specification over state signals deriving from dynamical systems such as \eqref{eq:single agent dynamics}  \cite{maler2004monitoring}. Specifically, let a set of \textit{predicate} functions $h: \mathbb{R}^{n} \rightarrow \mathbb{R}$, with level set
\begin{equation}\label{eq:general predicate set}
\mathcal{H} := \{\vec{x} \in \mathbb{X} \mid h(\vec{x}) \geq 0\}.
\end{equation}
The membership of a state $\vec{x}$ in the set $\mathcal{H}$ is logically encoded by the Boolean predicate $\mu^h(\vec{x}) = \mu(h(\vec{x})) :=
\bigl\{\begin{smallmatrix*}
\true &\text{if} \; h(\vec{x})\geq 0 \\
\false &\text{if} \; h(\vec{x})<0 ,
\end{smallmatrix*}$, such that $\mu^h(\vec{x})= \top \Leftrightarrow \vec{x} \in \mathcal{H}$. Using Backus–Naur notation, STL formulas are then recursively defined over a set of predicates according to the grammar: 
\begin{equation}\label{eq:general fragment}
\begin{gathered}
\phi::= \top \mid  \mu^{h} \mid \neg \phi \mid \phi_1 U_{[a,b]} \phi_2 \mid \phi_1 \land \phi_2,\\
\end{gathered}
\end{equation}
where $U$ is the temporal \textit{until} operator, with time interval $[a,b] \subset \mathbb{R}_{\geq 0}$, while $\neg$ and $\land$ represent logical negation and conjunction operators. The disjunction operator $\lor$  derives from these by De Morgans's laws. The temporal \textit{always} and \textit{eventually} operators derive from the \textit{until} as $ G_{[a,b]}\phi = \neg(\top U_{[a,b]} \neg \phi)$ and  $F_{[a,b]}\phi = \top U_{[a,b]} \phi$, respectively. We here consider tasks $\phi$ with bounded time domain \cite{farahani2018shrinking} i.e. formulas whose temporal operators have a bounded domain, for which we define the time horizon of a task $\phi$ as
\begin{subequations}\label{eq:horizon times}
\begin{align}
t_{hr}(\mu^{h}) &= 0, \; t_{hr}(\neg \phi) = t_{hr}(\phi), \\
t_{hr}(T_{[a,b]}\phi) &= t_{hr}(\phi) + b, \; T \in \{G,F\} \\
t_{hr}(\phi_1 \star \phi_2) &= \max\{t_{hr}(\phi_1), t_{hr}(\phi_2)\}, \; \star \in \{\land, \lor\},
\end{align}
\end{subequations}
where De Morgan's laws can be used to infer the time horizon for other operators. A simple example provides some intuition. 
\begin{example}\label{ex:example 1}
 Consider Fig. \ref{fig:RRT expansion example}. The predicate functions $h^c(\vec{x}):= 
\epsilon_r - \|
\vec{x}_c - \vec{x}\|$ and $h^i(\vec{x}):= 
\epsilon_r - \|
\vec{x}_i - \vec{x}\|$ represent a charging area (black battery in Fig. \ref{fig:RRT expansion example}) and a region of interest (warning sign in Fig. \ref{fig:RRT expansion example}), respectively.  The term $\epsilon_r>0$ is a positive radius, while $\vec{x}_c$ and $\vec{x}_i$ denote the position of the charging area and interest region in the workspace, respectively.  Dashed circles represent the level sets $\mathcal{H}^i$ and $\mathcal{H}^c$, where $\mu^{h^i}(\vec{x}) = \top$ and $\mu^{h^c} (\vec{x})= \top$, respectively. A task for the black drone (lower right in Fig. \ref{fig:RRT expansion example}) could be: ``In the next 10 minutes, always visit the charging station with intervals of at most 2 minutes and eventually visit the region of interest and stay there for at least 1 minute". In the STL formalism this can be written for example as $\phi = G_{[0,10]}F_{[0,2]}\mu^{h^c}(\vec{x}) \land F_{[0,10]}G_{[0,1]}\mu^{h^i}(\vec{x})$. In this work, we aim at designing a time-varying set $\mathcal{B}^{\phi}(t)$ (blue sets in Fig. \ref{fig:RRT expansion example}) such that the evolution of the drone state within the set $\mathcal{B}^{\phi}(t)$ guarantees the satisfaction of the task $\phi$. \hfill $\square$
\end{example}

 To characterize the satisfaction of a given STL task, we consider the STL quantitative semantics (see e.g. \cite[Def. 10]{fainekos2009robustness} and \cite[Sec. 2]{donze2013efficient}). Namely, let $(\zeta_x,t)\vDash \phi$ denote that the signal $\zeta_x : \mathbb{R}_{\geq 0} \rightarrow \mathbb{X}$ satisfies $\phi$ starting from the reference time $t\in \mathbb{R}_{\geq 0}$ and let the function  $\rho^{\phi}: \mathcal{F}(\mathbb{R}_{\geq 0},\mathbb{R}^n) \times \mathbb{R}_{\geq 0} \rightarrow  \mathbb{R}$ be recursively defined over a trajectory $\zeta_x$ as 
\begin{subequations}\label{eq:robust semantics}
\begin{align}
\rho^{\mu}(\zeta_x,t)&=h(\zeta_x(t)), \label{eq:predicate robust}\\
\rho^{\neg \phi}(\zeta_x,t)& = -\rho^{\phi}(\zeta_x,t), \label{eq:negation robust}\\
\rho^{F_{[a, b]} \phi}(\zeta_x,t)&=\max _{\tau \in t\oplus [a, b]} \{\rho^\phi\left(\zeta_x, \tau \right)\}, \label{eq:eventually robust}\\
\rho^{G_{[a, b]} \phi}(\zeta_x,t)&= \min _{\tau \in t\oplus[a, b]} \{\rho^\phi\left(\zeta_x, \tau\right)\},\label{eq:always robust}\\
\rho^{\phi_1 U_{[a,b]}\phi_2}(\zeta_x,t) &= \max_{\tau \in t\oplus[a,b]}\label{eq:until robust} \\ 
& \hspace{-1.5cm} \{ \min\{\rho^{\phi_2}(\zeta_x, \tau ), \min_{\tau' \in [t,\tau]} \rho^{\phi_1}(\zeta_x,\tau')\}\},\notag\\
\rho^{\phi_1 \wedge \phi_2}(\zeta_x,t)& = \min \left\{\rho^{\phi_1}(\zeta_x,t), \rho^{\phi_2}(\zeta_x,t)\right\},  \label{eq:conjunction robust} 
\end{align}
\end{subequations}
from which we know the semantic relation $\rho^{\phi}(\zeta_x,t) >0 \Rightarrow (\zeta_x,t)\vDash \phi$ \cite[Prop. 16]{fainekos2009robustness}. Note that the selection of the start time is arbitrary, and $t=0$ is selected as the convention hereafter without loss of generality. Differently from other semantics, quantitative semantics capture the \textit{degree of satisfaction} of a specification $\phi$. Specifically, for a given margin $r>0$, the signal $\zeta_x$ is said to \textit{robustly satisfy} $\phi$ with degree $r$, if $\rho^{\phi}(\zeta_x,0) \geq r>0$, which we denote as $(\zeta_x,0)\vDash_{r} \phi \equiv \rho^{\phi}(\zeta_x,0) \geq r$ such that
\begin{equation}
(\zeta_x,0)\vDash_{r} \phi \equiv \rho^{\phi}(\zeta_x,0) \geq r \Rightarrow (\zeta_x,0)\vDash \phi.
\end{equation}

In this work, we consider predicate functions of the form
\begin{equation}\label{eq:general form predicate}
\begin{aligned}
h(\vec{x})  := \min_{k\in [[n_h]]}\{\vec{d}_k^T \vec{x} +c_k\},\\
\end{aligned}
\end{equation}
for some $\vec{d}_k \in \mathbb{R}^{n}$, $c_k \in \mathbb{R}$ and $n_h\geq 1$. Thus, the set $\mathcal{H}$ in \eqref{eq:general predicate set} is a polyhedron since $h(\vec{x}) \geq 0 \Leftrightarrow D\vec{x} + \vec{c} \geq 0
$
where 
\begin{equation}
    D = \begin{bmatrix}
        \vec{d}_1^T \\
       \vdots \\
         \vec{d}_{n_h}^T \\
    \end{bmatrix}\in \mathbb{R}^{n_h\times n},\;
    \vec{c} = \begin{bmatrix}
       c_1 \\
       \vdots \\
      c_{n_h}\\
    \end{bmatrix} \in \mathbb{R}^{n_h}.
\end{equation}

Linear predicates, as per \eqref{eq:general form predicate}, are commonly considered in the STL literature for their computational tractability while allowing for rich types of specifications (cf. \cite{kurtz2022mixed,raman2015reactive,vasile2017sampling}). Moreover, we consider STL tasks expressed in the fragment  
\begin{subequations}\label{eq:subfragment}
    \begin{gather}
        \varphi ::= T_{[a,b]}\mu^{h} \mid T_{[a,b]}T'_{[a',b']} \mu^{h}, \label{eq:with time}\\
        \phi ::= \varphi \mid \varphi_1 \land \varphi_2,\label{eq:conjunction}\\
        \psi ::= \phi \mid \phi_1 \lor \phi_2 \label{eq:disjunction}
    \end{gather}
\end{subequations}
where $T,T' \in \{G,F\}$. Formulas $\varphi$, as per \eqref{eq:with time}, represent temporally extended specifications over a single predicate based on the \textit{always} or \textit{eventually} operators, or a composition of these. Note that, formulas of type $\mu^{h_1}U_{[a,b]}\mu^{h_2}$ can be expressed in \eqref{eq:subfragment} as $G_{[a,\tau]}\mu^{h_1}\land F_{[\tau,\tau]}\mu^{h_2}$ for some $\tau \in [a,b]$. Likewise, arbitrary nestings of the same temporal operator are part of \eqref{eq:with time} since $T_{[a_1,b_2]}T_{[a_2,b_3]} ...T_{[a_N,b_N]} \mu^{h} \equiv T_{[\sum_{n=1}^N a_n ,\sum_{n=1}^N b_n]} \mu^h$. Formulas $\phi$, as per \eqref{eq:conjunction}, represent conjunctions of formulas $\varphi$, as per \eqref{eq:conjunction}, to compose complex specifications (e.g. Example \ref{ex:example 1}), while formulas of type $\psi$, as per \eqref{eq:disjunction}, represent disjunctions of formulas $\phi$. Informally speaking, the formula $\psi = \lor_k\phi_k$ is satisfied if at least one of the formulas $\phi_k$ is satisfied.

While the fragment in \eqref{eq:subfragment} does not capture the full expressivity of the STL fragment in \eqref{eq:general fragment}, we focus on \eqref{eq:subfragment} for the following reasons. First, synthesizing trajectories that satisfy general STL specifications is an NP-hard problem, typically tackled using Mixed-Integer Linear Programming (MILP). Although MILP solvers are sound and complete, their computational demands are often prohibitive, even for formulas within the fragment \eqref{eq:subfragment}, making real-time planning infeasible. Regarding the absence of the negation operator in \eqref{eq:subfragment} — which is commonly used to enforce safety properties (e.g., ``always avoid region $\mathcal{H}$") — this limitation does not significantly reduce expressivity in our settings since safety requirements will be handled at the sampling level by rejecting trajectories entering unsafe regions of the workspace, e.g. obstacles. At the same time, we admittedly can handle the $\lor$ operator only at a high level of the formula such that formulas of type $F_{[a,b]}G_{[a',b']} (\mu^{h_1} \lor \mu^{h_2})$ are not within fragment \eqref{eq:subfragment} and we leave this extension as future work.

To summarize, compared to planning via MILP solvers (e.g. \cite{kurtz2022mixed}),  we trade off a reduction in expressivity,  for a simple and effective implementation that allows for real-time planning of trajectories satisfying STL tasks in the fragment \eqref{eq:subfragment}, from which a rich set of behaviors of common interest (e.g. rescue, exploration, and patrolling) can be obtained.
\subsection{Viability and forward invariance of set-valued maps}

As pointed out in the introduction, we propose to adopt a forward invariance perspective over the satisfaction of STL specifications from fragment \eqref{eq:subfragment} on the same line of \cite{lindemann_control_2019,charitidou2021barrier}. Specifically, we consider encoding STL tasks from the fragment \eqref{eq:subfragment} into a time-varying set derived as the level set of a non-smooth Control Barrier Function (CBF) \cite{glotfelter2019hybrid} of the form $\mathfrak{b}: \mathbb{R}^{n}\times [t_0,t_1] \rightarrow \mathbb{R}$, for some interval $[t_0,t_1] \subset \mathbb{R}_{\geq 0}$, with level set $\mathcal{B} : [t_0,t_1] \rightarrow 2^{\mathbb{X}}$ defined as
\begin{equation}\label{eq:general level set}
\mathcal{B}(t) := \{\vec{x}\in \mathbb{X} \mid \mathfrak{b}(\vec{x},t)\geq 0 \}.\\
\end{equation}
We consider the function $\mathfrak{b}$ to be Lipschitz continuous and concave on $\mathbb{R}^{n}$, while it is piece-wise linear on $[t_0,t_1]$, with a finite sequence of discontinuities $(s_i)_{i=1}^{n_s}$ such that $s_1 = t_0$ and $s_{n_s} = t_{1}$. We provide an analytical form of $\mathfrak{b}$ later in Section \ref{sec:stl tasks}, from which these assumptions will be clarified. For the time-varying set $\mathcal{B}(t)$, we often consider the left limit defined as $\lim_{\tau \rightarrow^{-} t }\mathcal{B}(\tau) = \{\vec{x}\in \mathbb{X} \mid \lim_{\tau \rightarrow^{-} t }\mathfrak{b}(\vec{x},\tau)\geq 0 \}$, which exists for all $t\in [t_0,t_1]$, as per piece-wise linearity of the function $\mathfrak{b}(\cdot, t)$, and we note that in general, $\lim_{\tau \rightarrow^{-} s_i }\mathcal{B}(\tau) \neq \mathcal{B}(s_i)$ due to the time discontinuity at each $s_i$. Moreover, $\mathcal{B}(t)$ is convex for every time $t\in [t_0,t_1]$, by the concavity of $\mathfrak{b}(\vec{x}, \cdot)$ and convexity of $\mathbb{X}$ \cite[Sec. 3.1.6]{boyd2004convex}. 

A pivotal aspect in our planning framework is the ability to synthesize dynamically feasible trajectories that evolve within a time-varying set of the form \eqref{eq:general level set}, which we ought to design in order to satisfy a given STL task. The notion of forward-invariance of set-valued maps is thus introduced.
\begin{definition}\label{def:forward invariance}
   A time-varying set $\mathcal{B}: [t_0,t_1] \rightarrow 2^{\mathbb{X}}$ is \textit{forward invariant} w.r.t \eqref{eq:single agent dynamics} over $[t_0,t_1]$, if for every $\vec{x}_0 \in \mathcal{B}(t_0)$, there exists an input signal $\zeta_u : [t_0,t_1] \rightarrow \mathbb{U}$, with corresponding solution $\zeta_x : [t_0,t_1] \rightarrow \mathbb{X}$ to \eqref{eq:single agent dynamics}, such that $\zeta_x(t_0)= \vec{x}_0$ and $\zeta_x(t)\in \mathcal{B}(t), \; \forall t \in [t_0,t_1]$.
\end{definition}

Due to the switching sequence $(s_i)_{i=1}^{n_s}$ induced by the function $\mathfrak{b}$, an essential precondition for forward invariance of the set $\mathcal{B}(t)$ is the notion of \textit{viability}.
\begin{definition}\label{def:vaiable}
    A time-varying set $\mathcal{B}: [t_0,t_1] \rightarrow 2^{\mathbb{X}}$ is \textit{viable} if $\mathcal{B}(t) \neq \emptyset$ and $\lim_{\tau \rightarrow^{-} t } \mathcal{B}(\tau) \subseteq  \mathcal{B}(t),\; \forall t \in [t_0,t_1]$.  
\end{definition}

The notion of viability ensures that a time-varying set $\mathcal{B}(t)$ is well behaved at the switching times $(s_i)_{i=1}^{n_s}$ according to the following intuition. If we let $\zeta_{x} : [t_0,t_1] \rightarrow \mathbb{X}$ be an absolutely continuous solution to \eqref{eq:single agent dynamics} such that $\lim_{\tau \rightarrow^{-} s_i} \zeta_x(\tau) \in \lim_{\tau \rightarrow^{-} s_i} \mathcal{B}(\tau)$ for any $s_i$ in the sequence $(s_i)_{i=1}^{n_s}$, then $\zeta_x(s_i) \in \mathcal{B}(s_i)$ even if $\lim_{\tau \rightarrow^{-} s_i} \mathcal{B}(\tau) \neq \mathcal{B}(s_i)$ such that $\zeta_x$ can not \textit{jump} outside $\mathcal{B}(t)$ at time $s_i$ (c.f. \cite[Def. 2]{glotfelter2019hybrid}, \cite[Sec. III]{lindemann_control_2019}). Hereafter, we term the function $\mathfrak{b}$ a \textit{candidate} Control Barrier Function (cCBF) over $[t_0,t_1]$ if its level set is viable. 

We next establish a sufficient condition under which $\mathcal{B}(t)$, as per \eqref{eq:general level set}, is guaranteed to be forward invariant w.r.t. the dynamics in \eqref{eq:single agent dynamics}. This problem was originally studied in  \cite{glotfelter2019hybrid} from which the next Theorem \ref{thm:forward invariance non smooth theorem} is inspired and adapted to the notation of our settings. Namely, let the generalized gradient $\partial \mathfrak{b}(\vec{x},t) : \mathbb{R}^{n}\times (t_0,t_1) \rightarrow 2^{\mathbb{R}^{n+1}}$ of $\mathfrak{b}$ be defined as
\begin{equation}\label{eq:generalized gradient}
\begin{aligned}
&\partial \mathfrak{b}(\vec{x},t) = \\
& co \left(\{ \lim_{i\rightarrow \infty} \nabla \mathfrak{b}(\vec{x}_i,t_i):(\vec{x}_i,t_i) \rightarrow (\vec{x},t), (\vec{x}_i,t_i) \in \mathcal{N} \cup\mathcal{N}_h \}\right),
\end{aligned}
\end{equation} 
where $\mathcal{N}$ is an arbitrary set of measure zero and $\mathcal{N}_h$ is the set of measure zero where $\mathfrak{b}$ is not differentiable \cite[Thm. 2.5.1]{clarke1990optimization}. Note that the generalized gradient is defined over open intervals, which is why we consider the open time interval $(t_0,t_1)$. Furthermore, let the \textit{extended} state-time dynamics 
\begin{equation}
\begin{gathered}
\dot{\vec{z}} = \bar{A} \vec{z} + \bar{B} \vec{u} + \bar{\vec{p}}, \\
\bar{A} = \begin{bmatrix} A &  0_{n\times 1} \\ 0_{1\times n}  &0 \end{bmatrix},\; \bar{B} =  \begin{bmatrix} B \\ 0_{1 \times m}\end{bmatrix}, \; \bar{\vec{p}} = \begin{bmatrix} \vec{p} \\ 1 \end{bmatrix} ,
\end{gathered}
\end{equation}
where $\vec{z} = [\vec{x}^T\; t]^T \in \mathbb{R}^{n+1}$ and let the weak set-valued generalized Lie derivative \cite[Eq. 3]{glotty} of $\mathfrak{b}$ be defined as
\begin{equation}\label{eq:generalized lie derivative}
\begin{aligned}
\mathcal{L}_{\mathfrak{b}}(\vec{x},t,\vec{u}) &:= \{ \vec{\nu}^T (\bar{A} \vec{z} + \bar{B} \vec{u} + \bar{\vec{p}}) \mid \;  \vec{\nu} \in \partial \mathfrak{b}(\vec{x},t)\}. \\
\end{aligned}
\end{equation}
\begin{theorem}\label{thm:forward invariance non smooth theorem}
Let $\mathfrak{b}: \mathbb{R}^{n} \times [t_0,t_1] \rightarrow \mathbb{R}$ be a cCBF and let the sequence $(s_i)_{i=1}^{n_s}$, with $s_1 = t_0$ and $s_{n_s}=t_1$, where $\mathfrak{b}$ is discontinuous. If for every $\vec{x}_0 \in \mathbb{X}$ there exists a measurable input signal $\zeta_u : [t_0,t_1] \rightarrow \mathbb{U}$ and a class-$\mathcal{K}$ function $\kappa : \mathbb{R} \rightarrow \mathbb{R}$, such that for all $i\in [[n_s-1]]$
\begin{equation}\label{eq:sufficient forward invariance condition}
\hspace{-0.25cm}\min \mathcal{L}_{\mathfrak{b}}(\zeta_x(t),t,\zeta_u(t)) \geq - \kappa(\mathfrak{b}(\zeta_x(t),t)), \forall t\in (s_{i},s_{i+1})
\end{equation}
with $\zeta_x:[t_0,t_1] \rightarrow \mathbb{X}$ being the solution to \eqref{eq:single agent dynamics} under $\zeta_u$ and such that 
$\zeta_x(0)=\vec{x}_0 $, then $\mathcal{B}(t)$ is forward invariant on $[t_0,t_1] $ and $\zeta_x(t) \in \mathcal{B}(t), \forall t\in [t_0,t_1]$.
\end{theorem}
\section{Problem Formulation}\label{sec:problem}
Consider the state set $\mathbb{X}$ and a set of obstacles $\mathcal{O}_k \subset \mathbb{X},\; k\in [[n_o]]$, for some $n_o \geq 0$ (e.g., black rectangles in Fig. \ref{fig:RRT expansion example}) such that a trajectory $\zeta_x : [t_0,t_1] \rightarrow \mathbb{X}$ is \textit{safe} if for all times $t \in [t_0,t_1]$ it holds $\zeta_x(t) \not\in  \cup_{k \in [[n_o]]}\mathcal{O}_k$ (i.e., no collisions occur). Then the problem we approach is formalized as follows:

\begin{problem}\label{prob:problem 1}
    Consider the dynamical system \eqref{eq:single agent dynamics} and an STL task $\phi$ as per \eqref{eq:conjunction} with maximum horizon $t_{hr}(\phi)$. Design an algorithm that returns a trajectory $\zeta_x : [0,t_{hr}(\phi)] \rightarrow \mathbb{X}$  such that $\zeta_x$ is safe and $(\zeta_x,0) \models_{r} \phi$ with robustness degree $r>0$. 
\end{problem}

To approach Problem \ref{prob:problem 1}, we develop the following steps. In the next Sections \ref{sec:stl tasks}-\ref{sec:forward invariance verification}, an algorithmic approach is defined to design a time-varying set $\mathcal{B}^{\phi}(t)$, in the form of \eqref{eq:general level set}, from a given task $\phi$, as per \eqref{eq:conjunction}, such that 1) the set $\mathcal{B}^{\phi}(t)$ is forward invariant as per Def. \ref{def:forward invariance}, (i.e., we show the existence of a control law that maintains systems \eqref{eq:single agent dynamics} within the set $\mathcal{B}^{\phi}(t)$ at every time) 2) any trajectory of system \eqref{eq:single agent dynamics} evolving in $\mathcal{B}^{\phi}(t)$ also satisfies $\phi$ with a given robustness $r>0$.

When considering the task disjunction $\psi = \lor_k \phi_k$, as per \eqref{eq:disjunction}, it is sufficient to only satisfy one of the tasks $\phi_k$ in order to satisfy $\psi$. Thus, we design a set $\mathcal{B}_k^{\phi}(t)$ for each $\phi_k$ in parallel, and the task achievable with highest robustness is selected to satisfy $\psi$, such that Problem \ref{prob:problem 1} naturally genealizes to the satisfaction of tasks of type $\psi$ as per \eqref{eq:disjunction}.

Eventually, in Section \ref{sec:rrt}, we generate trajectories that satisfy Problem \ref{prob:problem 1} via a modified implementation of RRT$^\star$. Namely, we enforce the STL task satisfaction by generating a tree of sampled trajectories that evolves within the previously designed time-varying sets, while safety is enforced by rejecting trajectories intersecting the obstacles $\mathcal{O}_k \subset \mathbb{X},\; k\in [[n_o]]$.
\section{From STL tasks to time-varying sets : Viability}\label{sec:stl tasks}
Leveraging the notion of \textit{parametric} Control Barrier Functions (CBF), in this section we propose an algorithmic approach to design a viable time-varying set $\mathcal{B}^{\phi}(t)$, encoding an STL task $\phi$, as per \eqref{eq:general level set},  as the level set of a parametric CBF. We start first by considering the construction approach for a single task $\varphi$, as per \eqref{eq:with time}, to then generalize to the conjunction $\phi$, as per \eqref{eq:conjunction}. When considering a task $\psi = \lor_k \phi_k$, obtained as the disjunction of multiple tasks $\phi_k$, we apply the same construction approach separately for each $\phi_k$.

\subsection{Parametric Control Barrier Functions}
For a task $\varphi$, let the family of \textit{parametric} CBFs $\mathfrak{b}^{\varphi}: \mathbb{R}^n \times [0,\beta] \times \Theta \rightarrow \mathbb{R}$ with $\beta \geq 0$, defined as \footnote{The notation  $\mathfrak{b}^{\varphi}(\cdot,\cdot|\vec{\vartheta})$ and $\mathcal{B}^{\varphi}(\cdot |\vec{\vartheta})$ applies to highlight the parametric dependence from the parameter vector $\vec{\vartheta}$.}
\begin{subequations}\label{eq:parametric forms}
\begin{align}
    &\mathfrak{b}^{\varphi}(\vec{x},t | \vec{\vartheta} ) = h(\vec{x}) +  \gamma^{\varphi}(t|\vec{\vartheta}) \label{eq:open form of the barriers},\\
    &\mathcal{B}^{\varphi}(t| \vec{\vartheta}) = \{ \vec{x}\in \mathbb{X} \mid \mathfrak{b}^{\varphi}(\vec{x}, t |\vec{\vartheta}) \geq 0 \}, \label{eq:parametric level set} 
\end{align}
\end{subequations}
where $h$ is the predicate associated to $\varphi$, as per \eqref{eq:general form predicate}, and the function $\gamma^{\varphi} : [0,\beta] \times \Theta \rightarrow \mathbb{R}$ is a continuous piece-wise linear function over the sequence of switches $(s_i)_{i=1}^{3} = (0,\alpha, \beta)$ with $\beta \geq \alpha \geq 0$ such that for $i\in \{1,2\}$ we have
\begin{equation}\label{eq:gamma general}
   \gamma^{\varphi}(t|\vec{\vartheta}) = e_i(\vec{\vartheta})\; t + g_i(\vec{\vartheta}), \; \forall t\in [s_{i},s_{i+1}], 
\end{equation}
with  
\begin{equation}\label{eq:coefficients of time compact}
\begin{array}{lll}
e_1(\vec{\vartheta}) = -\frac{\bar{\gamma}}{\alpha},  &  g_1(\vec{\vartheta}) = \bar{\gamma} -r, &t\in [s_1,s_{2}] =[0,\alpha],\\
e_2(\vec{\vartheta}) = 0,  &  g_2(\vec{\vartheta}) = - r, & t\in [s_2,s_{3}] = [\alpha,\beta].\\
\end{array}
\end{equation}
 An example of $\gamma^{\varphi}$ is given in Fig. \ref{fig:gamma general image}, where $\gamma^{\varphi}$ intuitively represents a functional encoding of the temporal operators $G_{[a,b]}$ and $F_{[a,b]}$, such that the reader should associate the interval $[\alpha,\beta]$ with the interval $[a,b]$ of the corresponding temporal operator.
The parameters of $\gamma^{\varphi}$ (and thus $\mathfrak{b}^{\varphi}$) are stacked in the vector $\vec{\vartheta} = [\bar{\gamma},  r ] ^T \in \Theta \subset \mathbb{R}^2 $. By letting the \textit{robust} level set of $h$ be defined as
\begin{equation}
\mathcal{H}^{r} := \{ \vec{x}\in \mathbb{X} \mid h(\vec{x}) -r \geq 0 \},
\end{equation}
we define $\Theta \subset \mathbb{R}^{2}$ as
\begin{equation}\label{eq:parameters bound}
\Theta := \left\{ \vec{\vartheta} \in \mathbb{R}^2 \; \mid  
\begin{array}{c}
     \bar{\gamma} \geq 0, \; r>0,\; \mathcal{H}^{r} \neq \emptyset
\end{array}\right\},
\end{equation}
from which we derive $\mathcal{H}^{r} \subset \mathcal{H}$, since $r>0$. Since $\gamma^{\varphi}$ is defined by two piece-wise linear sections, it will be useful to introduce the index map $\Upsilon : [0,\beta] \rightarrow \{1,2\}$ as
\begin{equation}\label{eq:upsilon}
\Upsilon(t) = \left\{\begin{array}{ll}
1 & \text{if} \; t\in [s_1,s_2)\\
2 & \text{if} \; t\in [s_2,s_3]\\
\end{array},\right.
\end{equation}
defining which of the two linear sections of $\gamma^{\varphi}$, a given time instant $t$ lies in (see right panel in Fig. \ref{fig:gamma general image}). At time instant $s_2$, we use the convention $\Upsilon(t)=2$ without loss of generality since $\gamma^{\varphi}$ is continuous at $s_2$. By the form of the predicate functions in \eqref{eq:general form predicate}, each $\mathfrak{b}^{\varphi}$ in \eqref{eq:parametric forms} is explicitly written as $\mathfrak{b}^{\varphi}(\vec{x},t|\vec{\vartheta}) = \min_{k\in [[n_h]]}\{ \vec{d}_k^T \vec{x} + c_k\} + \gamma^{\varphi}(t|\vec{\vartheta})$, and the level set $\mathcal{B}^{\varphi}(t|\vec{\vartheta})$ is a time-varying polyhedron, which is written in matrix form as
\begin{equation}\label{eq:polyhedral representation}
\mathcal{B}^{\varphi}(t|\vec{\vartheta}) = \left\{\vec{x} \in \mathbb{X} \mid  E^i(\vec{\vartheta}) \begin{bmatrix}
     \vec{x}  \\
     t 
\end{bmatrix} + \vec{g}^i(\vec{\vartheta}) \geq 0,\, i = \Upsilon(t) \right\},  
\end{equation}
where $E^i(\vec{\vartheta}) \in \mathbb{R}^{n_h \times (n+1)}$, and $\vec{g}^i(\vec{\vartheta}) \in \mathbb{R}^{n_h}$ are defined as
\begin{equation}\label{eq:normal vectors stack}
E^i(\vec{\vartheta}) = \begin{bmatrix}
     D & \vec{1}_{n_h} \cdot e_i(\vec{\vartheta}) \\
\end{bmatrix} ,\; 
\vec{g}^i(\vec{\vartheta}) = \vec{1}_{n_h} \cdot g_i(\vec{\vartheta}) + \vec{c},
\end{equation}
with index $i\in \{1,2\}$ indicating either of the two linear sections of $\gamma^{\varphi}$. As the next proposition shows, the set $\mathcal{B}^{\varphi}(t|\vec{\vartheta})$ is viable, by construction, for any $\vec{\vartheta} \in \Theta$.
\begin{proposition}\label{prop:important properties}
    For each function $\mathfrak{b}^{\varphi}$ it holds :
    \begin{enumerate}
        \item $\lim_{\tau \rightarrow^{-} t} \mathcal{B}^{\varphi}(\tau|\vec{\vartheta}) = \mathcal{B}^{\varphi}(t|\vec{\vartheta}),\; \forall t\in [0,\beta]$.
        \item $\mathcal{B}^{\varphi}(t|\vec{\vartheta})$ is viable.
        \item $\mathcal{B}^{\varphi}(t|\vec{\vartheta}) \supseteq \mathcal{B}^{\varphi}(\tau|\vec{\vartheta}),\; \forall \tau \in [t,\beta],\; \forall t\in [0,\beta]$.
    \end{enumerate}
\end{proposition}
\begin{proof}
1) By construction $\mathfrak{b}^{\varphi}$ is continuous since $h$ is continuous \cite[Prop. 7]{cortes2008discontinuous} and $\gamma^{\varphi}$ is continuous by construction, such that $\lim_{\tau \rightarrow^{-} t }\mathcal{B}^{\varphi}(\tau|\vec{\vartheta}) = \{\vec{x}\in \mathbb{X} \mid \lim_{\tau \rightarrow^{-} t }\mathfrak{b}(\vec{x},\tau|\vec{\vartheta})\geq 0 \} =  \{\vec{x}\in \mathbb{X} \mid \mathfrak{b}(\vec{x},t|\vec{\vartheta})\geq 0 \}= \mathcal{B}^{\varphi}(t|\vec{\vartheta})$. 2) From \eqref{eq:parameters bound}, the set $\mathcal{H}^{r} \neq \emptyset$, and, from the fact that $\bar{\gamma} \geq 0$, we have $\gamma^{\varphi}(t|\vec{\vartheta}) \geq -r, \forall t\in [0,\beta]$. Thus  $\mathfrak{b}(\vec{x},t) = h(\vec{x}) + \gamma^{\varphi}(t|\vec{\vartheta}) \geq h(\vec{x})-r$ for all $(\vec{x},t) \in \mathbb{X}\times [0,\beta]$, concluding that $\mathcal{B}^{\varphi}(t|\vec{\vartheta}) \supseteq \mathcal{H}^{r} \neq \emptyset, \; \forall t\in [0,\beta]$. Together with 1) this proves the viability of $\mathcal{B}^{\varphi}(t|\vec{\vartheta})$ as per Def. \ref{def:vaiable}. 3) Since $\gamma^{\varphi}(t|\vec{\vartheta})$ is monotonically decreasing on $[0,\beta]$, then for any $\vec{x} \in \mathbb{X}$ and $0  \leq t \leq \tau \leq \beta$, we have $\mathfrak{b}^{\varphi}(\vec{x},\tau | \vec{\vartheta} ) \leq \mathfrak{b}^{\varphi}(\vec{x},t | \vec{\vartheta})$. Thus $\mathcal{B}^{\varphi}(t|\vec{\vartheta}) \supseteq \mathcal{B}^{\varphi}(\tau|\vec{\vartheta}),\; \forall \tau\in [t,\beta], \; \forall t\in [0,\beta]$.
\end{proof}

\begin{figure}
    \centering
    \includegraphics[width=\linewidth]{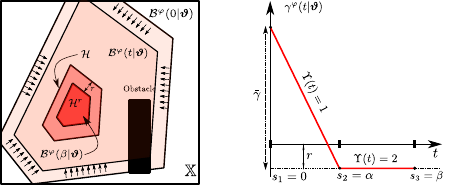}
    \caption{(left) Representation of the time-varying set $\mathcal{B}^{\varphi}(t| \vec{\vartheta})$ shrinking over the interval $[0,\beta]$. (right) Representation of the parametric function $\gamma^{\varphi}(t|\vec{\vartheta})$ with switching times $0,\alpha$ and $\beta$.} 
    \label{fig:gamma general image}
\end{figure}
%===============================================================================================================

Informally, Prop. \ref{prop:important properties} not only guarantees viability of $\mathcal{B}^{\varphi}(t|\vec{\vartheta})$, but also that the set is monotonically \textit{shrinking} toward the robust predicate level set $\mathcal{H}^{r}$ (left panel in Fig. \ref{fig:gamma general image}). This is expected from the fact that $\gamma^{\varphi}$ is monotonically decreasing for any choice of the parameters $\vec{\vartheta}$ in the set $\Theta$, as per \eqref{eq:parameters bound}. We highlight that the parameter $r$, representing the robustness of satisfaction for the task $\varphi$, should be chosen to be positive and such that $\mathcal{H}^{r} \neq \emptyset$ as per \eqref{eq:parameters bound}. This is indeed a necessary condition to have at least one state $\vec{x}\in \mathbb{X}$ where $h(\vec{x}) \geq r >0$, leading to the robust satisfaction of $\varphi$ as per the robust semantics in \eqref{eq:robust semantics}. From viability of $\mathcal{B}^{\varphi}(t|\vec{\vartheta})$, it follows directly that each barrier function $\mathfrak{b}^{\varphi}$ is a cCBF on $[0,\beta]$. 

We now present a set of rules to define the switching times $\alpha$ and $\beta$, depending on the temporal operators appearing in $\varphi$, as per \eqref{eq:with time}, such that trajectories evolving within $\mathcal{B}^{\varphi}(t)$ satisfy $\varphi$. For brevity, during the presentation of the next rules, we use the short-hand notation
\begin{equation}\label{eq:short hand fw}
\zeta_x \rhd_{[t_0,t_1]} \mathcal{B}(t)\;  \equiv \; \zeta_x(t) \in \mathcal{B}(t),\; \forall t\in [t_0,t_1].
\end{equation}

\subsection{Time-varying set encoding: single operator case}
First, the construction rules for formulas of type $\varphi = G_{[a,b]}\mu^{h}$ and $\varphi = F_{[a,b]}\mu^{h}$ are presented.
\begin{brule}\label{rule:eventually rule}
Let $\varphi = F_{[a,b]}\mu^{h}$ with $\mathfrak{b}^{\varphi}(\vec{x},t|\vec{\vartheta}) = \gamma^{\varphi}(t|\vec{\vartheta}) + h(\vec{x})$, where $\alpha \in [a,b]$ and $\beta \in [\alpha,b]$. Then $\zeta_x \rhd_{[0,\beta]} \mathcal{B}^{\varphi}(t|\vec{\vartheta}) \Rightarrow (\zeta_x,0) \vDash_{r} \varphi$.
\end{brule}
\begin{proof}
  By \eqref{eq:gamma general}, with $\alpha \in [a,b]$ and  $\beta \in [\alpha,b]$, we have $\gamma^{\varphi}(t|\vec{\vartheta}) = - r$ for all $t \in [\alpha,\beta]$, such that $\mathcal{B}^{\varphi}(t|\vec{\vartheta}) = \mathcal{H}^{r},\; \forall t \in [\alpha,\beta]$. Then $\zeta_x \rhd_{[0,\beta]} \mathcal{B}^{\varphi}(t|\vec{\vartheta}) \Rightarrow \zeta_x(t)\in \mathcal{B}^{\varphi}(t|\vec{\vartheta})= \mathcal{H}^{r},\; \forall t\in [\alpha,\beta] \Rightarrow  h(\zeta_x(t)) \geq r,\; \forall t\in [\alpha,\beta]$. Since $[\alpha,\beta] \subseteq [a,b]$, then by \eqref{eq:eventually robust} we have $\rho^{F_{[a,b]}}(\zeta_x,0) = \max_{\tau \in [a,b]} \{h(\zeta_x(\tau))\} \geq \max_{\tau \in [\alpha,\beta]} h(\zeta_x(t)) \geq r$, from which  $(\zeta_x,0)\vDash_{r} F_{[a,b]}\mu^{h} = \varphi$. 
\end{proof}

Thus, the set  $\mathcal{B}^{\varphi}(t|\vec{\vartheta})$ is intuitively designed to eventually converge to $\mathcal{H}^{r}$ for some freely chosen time $\alpha\in [a,b]$. Next, the rule for the \textit{always} operator is given where the set $\mathcal{B}^{\varphi}(t|\vec{\vartheta})$ is instead forced to converge to $\mathcal{H}^{r}$ at time $t=a$ and remain equal to $\mathcal{H}^{r}$ until $t=b$. 
\begin{brule}\label{rule:always formula} Let $\varphi = G_{[a,b]}\mu^{h}$ with $\mathfrak{b}^{\varphi}(\vec{x},t|\vec{\vartheta}) = \gamma^{\varphi}(t|\vec{\vartheta}) + h(\vec{x})$, where $\alpha=a,\; \beta=b$. Then $\zeta_x \rhd_{[0,\beta]} \mathcal{B}^{\varphi}(t|\vec{\vartheta}) \Rightarrow (\zeta_x,0) \vDash_{r}\varphi$.
\end{brule}
\begin{proof}
Noting that $\gamma^{\varphi}(t|\vec{\vartheta}) = -r, \forall t\in [a,b]$ from which $\mathcal{B}^{\varphi}(t|\vec{\vartheta}) = \mathcal{H}^r, \; \forall t\in [a,b]$, then the result follows, similarly to Rule \ref{rule:eventually rule}, applying the robust semantics \eqref{eq:always robust}.
\end{proof}

 We want to highlight that Rules \ref{rule:eventually rule} and \ref{rule:always formula} were originally presented in \cite{charitidou2021barrier,lindemann2018control,lindemann2025formal}, with different notation. On the other hand, the next two rules, involving nested operators, represent a novel contribution. 

\subsection{Time-varying set encoding: nested operator case}
\begin{brule}\label{rule:eventually always rule} Let $\varphi =  F_{[a,b]}G_{[a',b']}\mu^{h}$ with $\mathfrak{b}^{\varphi}(\vec{x},t|\vec{\vartheta}) = \gamma^{\varphi}(t|\vec{\vartheta}) + h(\vec{x})$, where $\alpha \in a' \oplus [a , b]$, $\beta = \alpha + (b'-a')$. Then $\zeta_x \rhd_{[0,\beta]} \mathcal{B}^{\varphi}(t|\vec{\vartheta}) \Rightarrow (\zeta_x,0) \vDash_{r}  \varphi$.
\end{brule}
\begin{proof}
    See Appendix \ref{app:rule 4}.
\end{proof}

Conceptually, we satisfy $\varphi = F_{[a,b]}G_{[a',b']}\mu^{h}$ similar to a task of type $\varphi = G_{[a',b']}\mu^{h}$ where now $\mathcal{B}^{\varphi}(t|\vec{\vartheta})$ converges to $\mathcal{H}^r$ over the shifted interval $\alpha\oplus [a',b']$ with $\alpha \in [a,b]$. Next, the rule for a task of type $\varphi = G_{[a,b]}F_{[a',b']}\mu^{h}$ is shown, following the result of the next proposition.
\begin{proposition}\label{prop:always eventually decomposition}
    Let $\varphi = G_{[a,b]}F_{[a',b']} \mu^{h}$ and $\tilde{\phi}= \land_{w=1}^{n_{f}} F_{[a_{w}, b_{w}]} \mu^{h}$ with $n_{f}\geq \lceil \frac{b-a}{b'-a'} \rceil$ and 
    \begin{equation}\label{eq:tau constraints}
        \begin{aligned}
        a_{w} &= a_{w-1} + \delta_w (b'-a'), \;  \delta_w \in \left[\frac{1}{n_f} \frac{b-a}{b'-a'},1 \right]\\
        b_w &= a_{w}
        \end{aligned}
    \end{equation}
    for all $w\in [[n_f]]$ and $a_0 = a+a'$. Then,  $(\zeta_x,0) \vDash_{r} \tilde{\phi} \Rightarrow (\zeta_x,0) \vDash_{r} \phi$.
\end{proposition}
\begin{proof}
    See Appendix \ref{app:always eventually proof}.
\end{proof}

By recalling that trajectories that robustly satisfying $\varphi = F_{[a,b]} \mu^{h}$ are those that visit the set $\mathcal{H}^{r}$ at some $t\in [a,b]$, then Prop. \ref{prop:always eventually decomposition} suggests that $\varphi = G_{[a,b]}F_{[a',b']} \mu^{h}$ can be satisfied by repeatedly revisiting the set $\mathcal{H}^{r}$ at least $n_f = \lceil \frac{b-a}{b'-a'} \rceil$ times. The next rule is a direct consequence of this result.
\begin{brule}\label{rule: always eventually} 
Let $\varphi = G_{[a,b]}F_{[a',b']}\mu^{h}$ and let the formulas $\varphi^w = F_{[a_w,b_w]}\mu^{h}$ with $w\in[[n_{f}]]$, $n_{f}\geq \lceil \frac{b-a}{b'-a'} \rceil$ and $a_w,b_w$ satisfying \eqref{eq:tau constraints}. For each $\varphi^w$ let the function $\mathfrak{b}_w^{\varphi}(\vec{x},t|\vec{\vartheta}_w)$ be constructed as per Rule \ref{rule:eventually rule} with level set $\mathcal{B}_w^{\varphi}(t|\vec{\vartheta}_w)$. Then $\zeta_x \rhd_{[0,\beta_w]} \mathcal{B}^{\varphi}_w(t|\vec{\vartheta}_w),\; \forall w\in [[n_{f}]] \Rightarrow (\zeta_x,0) \vDash_{r} \phi$.
\end{brule}
\begin{proof}
    Derives directly by Prop. \ref{prop:always eventually decomposition} and Rule \ref{rule:eventually rule}.
\end{proof}

Table \ref{tab:table of functions} summarises the requirements over the intervals $[\alpha,\beta]$ for each task type according to the construction Rules \ref{rule:eventually rule}-\ref{rule: always eventually} and note that for each task $\varphi$, we have $\beta \leq t_{hr}(\varphi)$.
\begin{table}[]
    \centering
    \begin{tabular}{llll}
    \toprule[2pt]
    \toprule[1pt]
        \# &Task ($\varphi$)  & Switch times & $t_{hr}(\varphi)$ \\
    \toprule[1pt]
        \ref{rule:eventually rule}&$F_{[a,b]}\mu^{h}$ & 
        $\alpha \in[a,b] , \beta \in [\alpha,b] $ & $b$\\
        \midrule
        \ref{rule:always formula}&$G_{[a,b]}\mu^{h}$ & $ \alpha=a , \beta=b   $& $b$\\
        \midrule
        \ref{rule:eventually always rule}&$F_{[a,b]}G_{[a',b']}\mu^{h}$   & $\begin{aligned} &\alpha = a' \oplus [a,b]  \\  &\beta = \alpha+ (b'-a')  \end{aligned} $& $b'+b$\\
        % \midrule
        % \ref{rule: always eventually}&  $\begin{aligned} &G_{[a,b]}F_{[a',b']}\mu^{h} \\&\text{satisfied by}  \\ &\land_{w=1}^{n_{f}} F_{[a_w,b_w]} \mu^{h} \end{aligned}$  & 
        % $\begin{aligned}
        % &\alpha_w = \beta_{w}, \forall w \in [1,n_{f}] \\
        %  &\alpha_{w} = \alpha_{w-1} + \delta_w \cdot (b'-a'), \\ 
        %  &\text{with}\\
        %  &\alpha_0 = a+a', \, n_{f}\geq \lceil \frac{b-a}{b'-a'}\rceil\\
        %   &\;  \delta_w \in \left[\frac{1}{n_f} \frac{b-a}{b'-a'},1 \right] 
        %   \end{aligned}$& $b'+b$\\ 
        \bottomrule[2pt]
    \end{tabular}
    \caption{Summary of Rules \ref{rule:eventually rule}-\ref{rule:eventually always rule}, with respective switching times. Each formula is converted into a parameteric cCBF $\mathfrak{b}^{\varphi}(\vec{x},t|\vec{\vartheta}) = \gamma^{\varphi}(t|\vec{\vartheta}) + h(\vec{x})$ where  $\gamma^{\varphi}(t|\vec{\vartheta})$ is defined as per \eqref{eq:gamma general} with switching times $\alpha$ and $\beta$ as summarized in the table. Rule \ref{rule: always eventually} is a special case of Rule \ref{rule:eventually rule} as per Prop. \ref{prop:always eventually decomposition} and is not reported explicitly.}
    \label{tab:table of functions}
\end{table}

\subsection{Time-varying set encoding: the conjunction case}\label{sec:task conjunction}
Consider now the conjunction $\phi = \land_{l=1}^{n_{\phi}} \varphi_{l}$ with $n_{\phi}\geq 1$. Specifically, let $\mathfrak{b}^{\varphi}_l$ and $\mathcal{B}_l^{\varphi}(t|\vec{\vartheta})$ be the cCBF and level set associated with each task $\varphi_{l}$ as per Rules \ref{rule:eventually rule}-\ref{rule: always eventually}\footnote{For ease of presentation, even if the formula $\varphi_{l} = G_{[a,b]}F_{[a',b']} \mu^{h}$ is technically associated with multiple functions $\mathfrak{b}^{\varphi}_{l,w}$ as per Rule \ref{rule: always eventually}, we still consider one barrier for each task since $\varphi_{l} = G_{[a,b]}F_{[a',b']} \mu^{h}$ is decomposed as a conjunction $\land_{w}\varphi_l^w$, which is then blended in the original conjunction $\phi = \land_l \varphi_l$.}. Then, for a given predicate $h_l(\vec{x}) = \min_{k\in [[n_h^l]]}\{\vec{d}_{k,l}^{T} + c_{k,l}\}$, each function  $\mathfrak{b}^{\phi}_l$ takes the form
\begin{equation}\label{eq:the fucking beta l}
\mathfrak{b}^{\phi}_l(\vec{x},t|\vec{\vartheta}_l) = \min_{k\in [[n_h^l]]} \{\vec{d}_{k,l}^{T} + c_{k,l}\} + \gamma^{\varphi}_l(t|\vec{\vartheta}_l) 
\end{equation}
with 
$$
\gamma_l^{\varphi}(t|\vec{\vartheta}_l) = e_{i,l}(\vec{\vartheta}_l) \cdot t + g_{i,l}(\vec{\vartheta}_l), \; \forall t\in [s^l_{i},s^l_{i+1}],
$$
as per \eqref{eq:gamma general}, where $\vec{\vartheta}_l = [\bar{\gamma}_l, r_l]^T \in \mathbb{R}^2 $ and $(s_i^l)_{i=1}^{3}=(0,\alpha_l,\beta_l)$ is the switching sequence of task $\varphi_l$. The next Lemma \ref{lemma:conjunction} relates the satisfaction of $\phi$ with the barrier function 
\begin{equation}\label{eq:super min barrier}
\begin{aligned}
\mathfrak{b}^{\phi}(\vec{x},t|\vec{\theta}) := \min_{l \in \mathcal{I}^{\phi}(t)} \{\mathfrak{b}_l^{\varphi}(\vec{x},t|\vec{\vartheta}_l)\}
\end{aligned}
\end{equation}
and corresponding time-varying level set
\begin{equation}\label{eq:conjunction set}
\mathcal{B}^{\phi}(t|\vec{\theta}) := \{\vec{x} \in \mathbb{X} \mid \mathfrak{b}^{\phi}(\vec{x},t|\vec{\theta}) \geq 0 \} = \hspace{-0.2cm} \bigcap_{l\in \mathcal{I}^{\phi}(t)} \hspace{-0.2cm} \mathcal{B}_l^{\varphi}(t|\vec{\vartheta}_l),
\end{equation}
where $\vec{\theta} = [\vec{\vartheta}_{1}^T, \ldots \vec{\vartheta}_{n_{\phi}}^T]^T \in \bar{\Theta}$, $\bar{\Theta} = \bigtimes_{l=1}^{n_{\phi}} \Theta_l$ and the switch map  $\mathcal{I}^{\phi}: \mathbb{R}_{\geq 0} \rightarrow 2^{\mathbb{N}}$ is defined as   
\begin{equation}\label{eq:switch map}
\mathcal{I}^{\phi}(t) := \{ l\in [[n_{\phi}]] \mid \beta_{l} > t\}.
\end{equation}
Namely, let, without loss of generality, the time instants $\beta_l$ be gathered into an ordered sequence $(\beta_l)_{l=0}^{n_{\phi}}$, where we set the convention $\beta_0 = 0$ and such that $\beta_l \geq \beta_{l-1}$ for all $l\in [[n_{\phi}]]$. Then the map $\mathcal{I}^{\phi}(t)$ is interpreted as a \textit{switch} map, removing the task of index $l$ when $\varphi_l$ has already been satisfied after the time $\beta_l$ has passed (see Fig. \ref{fig:sets limits figure}). Indeed, the cardinality of $\mathcal{I}^{\phi}(t)$ reduces over time with $\mathcal{I}^{\phi}(t) \subseteq \mathcal{I}^{\phi}(\tau)$ when $t \geq \tau$, while $\mathcal{I}^{\phi}(t)$ is constant in the half-open intervals $[\beta_{l-1},\beta_l)$. Noting that at the exact time $\beta_l$ the map $\mathcal{I}^{\phi}(\beta_l)$ does not contain the index $l$, then in general $\lim_{t\rightarrow^{-} \beta_l} \mathcal{B}^{\phi}(t|\vec{\theta}) \neq \mathcal{B}^{\phi}(\beta_l|\vec{\theta})$ due to the remotion of at least one set $\mathcal{B}_l^{\varphi}(t|\vec{\vartheta}_l)$ from the conjunction set \eqref{eq:conjunction set} at time $\beta_l$ (see Fig.  \ref{fig:sets limits figure}). Additionally, viability of $\mathcal{B}_l^{\varphi}(\beta_l|\vec{\vartheta}_l)$, for all $l\in [[n_{\phi}]]$, does not necessarily imply viability of the conjunction set $\mathcal{B}^{\phi}(t|\vec{\theta})$. As shown next, non-emptiness of $\mathcal{B}^{\phi}(t|\vec{\theta})$ is sufficient to prove its viability and trajectories evolving in $\mathcal{B}^{\phi}(t|\vec{\theta})$ also robustly satisfy $\phi$.
\begin{figure}[t]
    \centering
    \includegraphics[width=1.\linewidth]{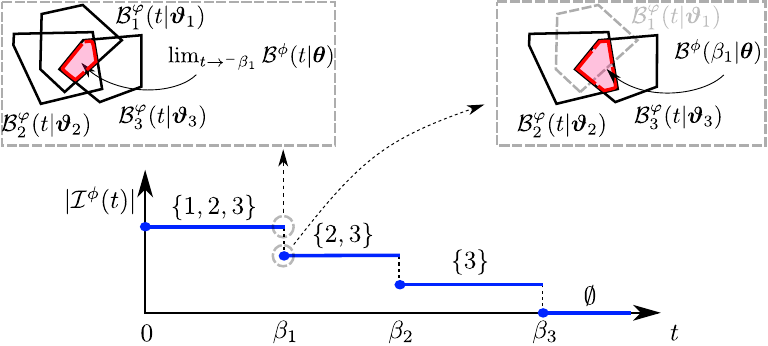}
    \caption{Example of limit of the set $\mathcal{B}^{\phi}(t|\vec{\theta})  = \cap_{l\in\mathcal{I}^{\phi}(t)} \mathcal{B}_l^{\varphi}(t|\vec{\vartheta}_l)$ at the switching time $\beta_1$ for a conjunction of three tasks $\varphi_1$,$\varphi_2$ and $\varphi_3$. The blue graph shows the time progression of the switch map $\mathcal{I}^{\phi}(t)$.}
    \label{fig:sets limits figure}
\end{figure}

\begin{lemma}\label{lemma:conjunction}
    If $\mathcal{B}^{\phi}(t|\vec{\theta}) \neq \emptyset, \; \forall t\in [0,\beta_{n_{\phi}}]$, with 
    \begin{equation}\label{eq:maximum satisfaction time}
    \beta_{n_{\phi}} = \max_{l\in [[n_{\phi}]]}\{\beta_l\},
    \end{equation}
    then $\mathcal{B}^{\phi}(t|\vec{\theta})$ is viable. Moreover, any solution $\zeta_x :[0,\beta_{n_{\phi}}] \rightarrow \mathbb{X}$ to \eqref{eq:single agent dynamics} such that $\zeta_x \rhd_{[0,\beta_{n_{\phi}}]} \mathcal{B}^{\phi}(t|\vec{\theta})$ implies $(\zeta_x,0) \models_r \phi$ with $r=\min_{l\in[[n_{\phi}]]}\{r_l\}$.
\end{lemma}
\begin{proof}
    Concerning viability, by Prop. \ref{prop:important properties}, each $\mathcal{B}_l^{\varphi}(t|\vec{\vartheta}_l)$ is viable over $[0,\beta_l]$ and also $\lim_{\tau \rightarrow^{-} t }\mathcal{B}_l^{\varphi}(\tau|\vec{\vartheta}_l) = \mathcal{B}_l^{\varphi}(t|\vec{\vartheta}_l),\; \forall t \in [0,\beta_l]$. Thus
    \begin{subequations}\label{eq:limit of big set}
    \begin{align}
    &\lim_{\tau \rightarrow^{-} t }\mathcal{B}^{\phi}(\tau|\vec{\theta}) =  \lim_{\tau \rightarrow^{-} t } \left(\bigcap_{l\in \mathcal{I}^{\phi}(\tau)} \mathcal{B}_l^{\varphi}(\tau|\vec{\vartheta}_l)\right) =  \\ 
    & = \bigcap_{l\in \lim_{\tau \rightarrow^{-} t} \mathcal{I}^{\phi}(\tau)} \lim_{\tau \rightarrow^{-} t } \mathcal{B}_l^{\varphi}(\tau|\vec{\vartheta}_l)  \label{eq:the passage}\\
   &=  \bigcap_{l\in \lim_{\tau \rightarrow^{-} t} \mathcal{I}^{\phi}(\tau)} \mathcal{B}_l^{\varphi}(t|\vec{\vartheta}_l) \quad  \text{(by Prop. 1 - (1) )}\\
   &\subseteq \bigcap_{l\in \mathcal{I}^{\phi}(t)} \mathcal{B}_l^{\varphi}(t|\vec{\vartheta}_l)\quad (\text{by}\; \mathcal{I}^{\phi}(t)\subseteq \;   \mathcal{I}^{\phi}(\tau), \; t \geq \tau ) \\
   &= \mathcal{B}^{\phi}(t|\vec{\theta}),
    \end{align}
    \end{subequations}
    for all $t \in [0,\beta_{n_{\phi}}]$. From this fact and letting $\mathcal{B}^{\phi}(t|\vec{\theta}) \neq \emptyset,\; \forall t\in [0,\beta_{n_{\phi}}]$, as per the lemma statement, then $\mathcal{B}^{\phi}(t|\vec{\theta})$ is viable as per Def. \ref{def:vaiable}. Consider now the second statement and let any index $l \in [[n_{\phi}]]$. The condition $\zeta_x \rhd_{[0,\beta_{n_{\phi}}]} \mathcal{B}^{\phi}(t|\vec{\theta})$ implies, by definition of $\mathcal{I}^{\phi}(t)$, that $\zeta_x(t) \in \mathcal{B}_{l}^{\varphi}(t|\vec{\vartheta}_{l})$ for all $t\in [0, \beta_{l})$. Since $\zeta_x: [0,\beta_{n_{\phi}}] \rightarrow \mathbb{X}$ is absolutely continuous, and $\mathcal{B}_{l}^{\varphi}(t|\vec{\vartheta}_{l})$ is well defined over the closed interval $[0,\beta_{l}]$, then  $\lim_{t\rightarrow^{-} \beta_{l} } \zeta_x(t) \in \lim_{t \rightarrow^{-} \beta_{l}} \mathcal{B}_{l}^{\varphi}(t|\vec{\vartheta}_{l})$ and thus, by Prop \ref{prop:important properties} point 1), $\zeta_x(\beta_{l}) \in \mathcal{B}_{l}^{\varphi}(\beta_{l}|\vec{\vartheta}_{l})$, even if $l\not\in \mathcal{I}^{\phi}(\beta_l)$. Hence, $\zeta_x(t) \in \mathcal{B}_{l}^{\varphi}(t|\vec{\vartheta}_{l})$ for all $t\in [0, \beta_{l}]$ and, by Rules \ref{rule:eventually rule}-\ref{rule: always eventually}, $\rho^{\varphi_{l}}(\zeta_x,0) \geq r_{l}$. Thus we conclude $\zeta_x \rhd_{[0,\beta_{n_{\phi}}]} \mathcal{B}^{\phi}(t|\vec{\theta}) \Rightarrow  \rho^{\varphi_{l}}(\zeta_x,0) \geq r_l,\; \forall l \in [[n_{\phi}]] \Rightarrow \rho^{\phi}(\zeta_x,0) = \min_{l \in [[n_{\phi}]]}\{\rho^{\varphi_{l}}(\zeta_x,0)\} \geq r$. Thus,  $(\zeta_x,0) \models_r \phi$ by \eqref{eq:conjunction robust}.
\end{proof}

Hence, as a result of Lemma \ref{lemma:conjunction}, to prove viability of $\mathcal{B}^{\phi}(t|\vec{\theta})$ it is sufficient to check for non-emptiness over the interval $[0,\beta_{n_{\phi}}]$, which can be done by checking a finite number of inclusions as shown next.
\begin{proposition}\label{prop:discrete viability prop}
    The set $\mathcal{B}^{\phi}(t|\vec{\theta})$ is viable if for all $(\beta_l)_{l=1}^{n_{\phi}}$, there exists $\vec{\xi}_l \in \mathbb{X}$ such that $\vec{\xi}_l \in \lim_{t\rightarrow^{-} \beta_l} \mathcal{B}^{\phi}(t|\vec{\theta})$.
\end{proposition}
\begin{proof}
    Given in Appendix \ref{proof of discrete viability prop}.
\end{proof}

Since $\mathcal{B}^{\phi}(t|\vec{\theta})$ is an intersection of polyhedrons, then the condition $\vec{\xi}_l \in \lim_{t\rightarrow^{-} \beta_l} \mathcal{B}^{\phi}(t|\vec{\theta})$ is equivalent to a set of linear inequalities. Namely, recall \eqref{eq:the passage} and recall that $\mathcal{I}^{\phi}(t)$ is constant over the interval $[\beta_{l-1}, \beta_l)$ such that $\lim_{\tau \rightarrow^{-} \beta_l} \mathcal{I}^{\phi}(\tau) = \mathcal{I}^{\phi}(\beta_{l-1})$, for which we have
\begin{equation}
\lim_{t\rightarrow^{-} \beta_l} \mathcal{B}^{\phi}(t|\vec{\theta}) = \bigcap_{l\in \mathcal{I}^{\phi}(\beta_{l-1})} \mathcal{B}_l^{\varphi}(\beta_{l}|\vec{\vartheta}_l).
\end{equation}
Noting that each $\mathcal{B}_l^{\varphi}(t|\vec{\vartheta}_l)$ has the matrix form  \eqref{eq:polyhedral representation}, then the condition $\vec{\xi}_l \in \lim_{t\rightarrow^{-} \beta_l} \mathcal{B}^{\phi}(t|\vec{\theta})$ in Prop. \ref{prop:discrete viability prop} is equivalent to the set of linear inequalities 
\begin{equation}
E^i_{\tilde{l}}(\vec{\vartheta}) \; \begin{bmatrix} \vec{\xi}_l \\ \beta_l\end{bmatrix} + \vec{g}_{\tilde{l}}^i(\vec{\vartheta}) \geq 0,\; \begin{array}{c}
i = \Upsilon_{\tilde{l}}(\beta_l),\\
\forall \tilde{l}\in \mathcal{I}^{\phi}(\beta_{l-1}),
\end{array}
\end{equation}
where  $\Upsilon_{\tilde{l}}(\beta_l)=1$ or $\Upsilon_{\tilde{l}}(\beta_l)=2$ depending on which linear section of the barrier $\mathfrak{b}^{\varphi}_{\tilde{l}}$, with $\tilde{l}\in \mathcal{I}^{\phi}(\beta_{l-1})$, the time $\beta_l$ is found, as per \eqref{eq:upsilon}.

\section{From STL tasks to time-varying sets : Forward Invariance}\label{sec:forward invariance verification}
Now that we have analyzed the viability properties of the set $\mathcal{B}^{\phi}(t|\vec{\theta})$, we want to focus on forward invariance. Namely, let the set $\mathcal{B}^{\phi}(t|\vec{\theta})$ be defined by an appropriate selection of the switching times from Rules \ref{rule:eventually rule}-\ref{rule: always eventually}. We want to find an optimal assignment of the parameters $\vec{\theta} = [\vec{\vartheta}_1^T \ldots \vec{\vartheta}^T_{n_{\phi}}]^T\in \bar{\Theta}$, such that $\mathcal{B}^{\phi}(t|\vec{\theta})$ is forward invariant and the robustness of satisfaction of the task $\phi = \land^{n_{\phi}}_{l=1}\varphi_l$ is maximized, by considering the general optimization problem
\begin{equation}\label{eq:very general optimization program}
\begin{aligned}
\min_{\vec{\theta}} \; \sum_{l=1}^{n_{\phi}} -r_l \quad \text{s.t.} \quad \mathfrak{g}(\vec{\theta}) \geq 0,
\end{aligned}
\end{equation}
where $\mathfrak{g} : \bar{\Theta} \rightarrow \mathbb{R}^{n_g}$ is a set of constraints on $\vec{\theta}$ such that the optimal set of parameters $\vec{\theta}^*$ solving \eqref{eq:very general optimization program} guarantees $\mathcal{B}^{\phi}(t|\vec{\theta}^*)$ is viable and forward invariant while  maximizing the robustness of satisfaction for $\phi$. In this section, we clarify how the constraints in $\mathfrak{g}$ are defined starting first by showing constraints under which a single set $\mathcal{B}^{\varphi}(t|\vec{\vartheta})$ can be made forward invariant, to then generalize to the conjunction set $\mathcal{B}^{\phi}(t|\vec{\theta})$, based on the result of Theorem \ref{thm:forward invariance non smooth theorem}. 
\subsection{Enforcing forward invariance : single task case}\label{subsec:single task}
For convenience, we temporarily omit the dependence from the parameters $\vec{\vartheta}$ and drop the index $l$ of each cCBF $\mathfrak{b}^{\varphi}$ in the derivations. We reintroduce full notation in the presentation of the main results. Namely, recall that for a single task $\varphi$ we have the associated cCBF $\mathfrak{b}^{\varphi}(\vec{x},t)= \min_{k\in [[n_h]]}\{\vec{d}_k^T \vec{x} +c_k\} + \gamma^{\varphi}(t),$ as per \eqref{eq:parametric forms}-\eqref{eq:gamma general} with switching sequence $(s_i)_{i=1}^{3} = (0,\alpha, \beta)$ induced by the function $\gamma^{\varphi}$, and recall the map $\Upsilon(t)$ identifying in which of the two linear sections of $\gamma^{\varphi}$ a given time instant $t$ is found, as per \eqref{eq:upsilon} (see Fig. \ref{fig:gamma general image}).

To enforce the forward invariance of the set $\mathcal{B}^{\varphi}(t)$ over the interval $[0,\beta]$ we leverage the result of Theorem \ref{thm:forward invariance non smooth theorem} for which we show the existence of a control input $\zeta_u : [0,\beta] \rightarrow \mathbb{U}$ satisfying the condition
\begin{equation}\label{eq:forward invariance varphi}
\min \mathcal{L}_{\mathfrak{b}^{\varphi}}(\zeta_x(t),t,\zeta_u(t)) \geq - \kappa(\mathfrak{b}^{\varphi}(\zeta_x(t),t)), \; \forall t\in (s_i,s_{i+1}),
\end{equation}
for both intervals $(s_i,s_{i+1}),i\in \{1,2\}$ and for every initial condition $\vec{x}_0 \in \mathcal{B}^{\varphi}(0)$, with $\zeta_x : [0,\beta]\rightarrow \mathbb{X}$ being the solution to \eqref{eq:single agent dynamics} under $\zeta_u$ such that $\zeta_x(0)= \vec{x}_0$. To this purpose, we further analyze the value of the Lie derivative over each interval $(s_i,s_{i+1})$. Particularly, since for all $t\in (s_i,s_{i+1})$ the function $\gamma^{\varphi}$ is differentiable, the generalized gradient of $\mathfrak{b}^{\varphi}$ over $(s_i,s_{i+1})$, $i\in \{1,2\}$ can be computed applying \cite[Prop. 6(iii), Prop. 7(iii)]{cortes2008discontinuous} as
\begin{equation}\label{eq:task generalized gradient}
\partial^{i}\mathfrak{b}^{\varphi}(\vec{x},t) = co(\{ [\vec{d}_k^T\; \vec{e}_i]^T\in \mathbb{R}^{n+1} \mid k\in \mathcal{A}(\vec{x},t) \}), 
\end{equation}
where
\begin{equation}\label{eq:active set map for barrier} 
 \mathcal{A}(\vec{x},t) = \{ k\in [[n_h]] \mid \vec{d}_k^T \vec{x} + c_k + \gamma^{\varphi}(t) = \mathfrak{b}^{\varphi}(\vec{x},t) \},
\end{equation}
is the set of \textit{active} components of $\mathfrak{b}^{\varphi}(\vec{x},t)$ defining the components exactly equal to the minimum in the definition of $\mathfrak{b}^{\varphi}$. After replacing \eqref{eq:task generalized gradient} into \eqref{eq:generalized lie derivative}, we also have that the weak set-valued Lie derivative of $\mathfrak{b}^{\varphi}$ over the intervals $(s_i,s_{i+1}), i\in \{1,2\}$, takes the convex hull representation\footnote{This is derived replacing \eqref{eq:task generalized gradient} into \eqref{eq:generalized lie derivative} and recalling that for any set of vectors $\{\vec{\nu}_j\}\subset \mathbb{R}^n$ and for a given vector $\vec{c} \in \mathbb{R}^n$, then it holds $\{\vec{\nu}^T \vec{c} \mid \; \vec{\nu} \in co(\{ \vec{\nu}_j \})   \} = co(\{\vec{c}^T \vec{\nu}_j\})$ by the definition of convex hull.}
\begin{equation}
\begin{array}{l}\label{eq:convex hull lie derivative representation}
   \mathcal{L}^i_{\mathfrak{b}^{\varphi}}(\vec{x},\vec{u},t) = co \left(\{[\vec{d}_k^T\; e_i] (\bar{A}\vec{z} + \bar{B}\vec{u} + \bar{\vec{p}})\mid k\in \mathcal{A}(\vec{x},t)\}\right ),
\end{array}
\end{equation} 
where we remind that $\vec{z} = [\vec{x}^T\; t]^T$. We thus further consider the control signal $\zeta_u : [0,\beta] \rightarrow \mathbb{U}$ given by
\begin{subequations}\label{eq:standard control input}
\begin{gather}
    \zeta_u(t) = \text{argmin}_{\vec{u} \in \mathbb{U}} \; \|\vec{u}\|  \quad  \text{s.t.} \\
     \min \mathcal{L}^i_{\mathfrak{b}^{\varphi}}(\vec{x},t,\vec{u}) \geq - \kappa(\mathfrak{b}^{\varphi}(\vec{x},t)),\; i= \Upsilon(t), \label{eq:condition single shot}\\
     \vec{x} = \zeta_{x}(t),
    \end{gather}
\end{subequations}
for all $t\in (s_i,s_{i+1}), i\in \{1,2\}$, where the value of $\zeta_u(s_i)$ for $\; i=\{1,2,3\}$ can potentially take any bounded value (since this happens on a set of measure zero), even if in practice we set $\zeta_u(s_i) = \lim_{t\rightarrow s_i} \zeta_u(s_i)$. Note that the map $\Upsilon(t) =i, \; \forall t\in (s_i,s_{i+1})$ is constant by definition over the intervals $(s_i,s_{i+1})$, as per \eqref{eq:upsilon}. By the fact that, $\mathcal{L}_{\mathfrak{b}^{\varphi}}(\vec{x},t,\vec{u})  = \mathcal{L}^i_{\mathfrak{b}^{\varphi}}(\vec{x},t,\vec{u}), \; \forall t\in (s_i,s_{i+1})$, it follows directly that if $\zeta_u$ in \eqref{eq:standard control input} is well defined (i.e. \eqref{eq:standard control input} has a solution) and it is measurable over $[0,\beta]$, then $\zeta_u$ satisfies \eqref{eq:forward invariance varphi} and thus $\mathcal{B}^{\varphi}(t)$ is forward invariant over $[0,\beta]$ by Thm. \ref{thm:forward invariance non smooth theorem}. We thus further analyze the feasibility of \eqref{eq:standard control input}. In the presentation, we assume that the control input \eqref{eq:standard control input} is measurable, which is a standard assumption when working with non-smooth dynamical systems (see e.g. \cite[Sec. III(B)]{glotfelter2019hybrid}), as the measurability assumption is violated only in pathological cases, which rarely occur in practice. 

\begin{proposition}\label{prop:preparatory convex hull representation}
Let an interval $(s_i,s_{i+1}),i\in \{1,2\}$. For any $(\vec{x},t) \in \mathbb{X} \times (s_{i},s_{i+1})$ there exists $\vec{u} \in \mathbb{U}$ satisfying 
\begin{equation}\label{eq:vertices condition}
[\vec{d}_k^T\; e_i] (\bar{A}\vec{z} + \bar{B}\vec{u} + \bar{\vec{p}}) \geq - \kappa(\mathfrak{b}^{\varphi}(\vec{x},t)), \; \forall k \in  \mathcal{A}(\vec{x},t)
\end{equation}   
if and only if $\min \mathcal{L}^i_{\mathfrak{b}^{\varphi}}(\vec{x},t,\vec{u}) \geq - \kappa(\mathfrak{b}^{\varphi}(\vec{x},t)).$
\end{proposition}
\begin{proof}
    Given in Appendix \ref{proof:preparatory convex hull representation}.
\end{proof}

Proposition \ref{prop:preparatory convex hull representation} suggests that \eqref{eq:condition single shot} is equivalent to a set of linear inequalities, which we can use to verify the feasibility of \eqref{eq:condition single shot} over each interval $(s_i,s_{i+1})$, by checking the existence of a control input $\vec{u}\in \mathbb{U}$ satisfying \eqref{eq:vertices condition} for each $(\vec{x},t) \in \mathbb{X} \times (s_{i},s_{i+1})$. However, recalling that $\mathfrak{b}^{\varphi}$ depends, by definition, on the parameters $\vec{\vartheta}$, we have that also $\mathcal{A}(\vec{x},t)=\mathcal{A}(\vec{x},t|\vec{\vartheta})$, which is known before setting $\vec{\vartheta}$ (note indeed that in \eqref{eq:active set map for barrier} we have that $\gamma^{\varphi}(t) = \gamma^{\varphi}(t|\vec{\vartheta})$). We thus resort to a sufficient condition that allows for checking the feasibility of \eqref{eq:condition single shot} without recurring to the knowledge of $\mathcal{A}(\vec{x},t)$.
\begin{proposition}\label{prop:the piece of the puzzle that will make me the king}
Let an interval $(s_i,s_{i+1}),i\in \{1,2\}$. For any $(\vec{x},t) \in \mathbb{X} \times (s_{i},s_{i+1})$, if there exists $\vec{u} \in \mathbb{U}$ such that 
\begin{equation}\label{eq:vertices condition weaker}
[\vec{d}_k^T\; e_{i}] (\bar{A}\vec{z} + \bar{B}\vec{u} + \bar{\vec{p}}) \geq - \kappa( \vec{d}_k^T\vec{x} + c_k + \gamma^{\varphi}(t)), \forall k \in [[n_{h}]]
\end{equation}  
then also $\min \mathcal{L}^i_{\mathfrak{b}^{\varphi}}(\vec{x},t,\vec{u}) \geq - \kappa(\mathfrak{b}^{\varphi}(\vec{x},t)).$
\end{proposition}
\begin{proof}
    We prove that \eqref{eq:vertices condition weaker} implies \eqref{eq:vertices condition}. The result then follows by Prop. \ref{prop:preparatory convex hull representation}. For any $i\in \{1,2\}$, select $(\vec{x},t)\in \mathbb{X} \times (s_{i},s_{i+1})$. By \eqref{eq:active set map for barrier}, we have $\vec{d}_k^T \vec{x}+ c_k + \gamma^{\varphi}(t) = \mathfrak{b}^{\varphi}(\vec{x},t)$ for all $k\in \mathcal{A}(\vec{x},t)$ and $\vec{d}_k^T \vec{x}+ c_k + \gamma^{\varphi}(t) > \mathfrak{b}^{\varphi}(\vec{x},t)$ for all $k\not\in \mathcal{A}(\vec{x},t)$. Thus  \eqref{eq:vertices condition weaker} implies \eqref{eq:vertices condition} since $[\vec{d}_k^T\; e_{i}] (\bar{A}\vec{z} + \bar{B}\vec{u}+\bar{\vec{p}}) \geq - \kappa(\vec{d}_k^T \vec{x}+ c_k + \gamma^{\varphi}(t)) =- \kappa(\mathfrak{b}(\vec{x},t)) , \; \forall k \in \mathcal{A}(\vec{x},t)$.
\end{proof}

Differently from \eqref{eq:vertices condition}, condition \eqref{eq:vertices condition weaker} does not require knowledge of the map $\mathcal{A}(\vec{x},t)$ at the cost of introducing extra constraints at each $(\vec{x},t)$. 
% The relation between \eqref{eq:vertices condition} and \eqref{eq:vertices condition weaker} is practically determined by the selection of the class-$\mathcal{K}$ function $\kappa: \mathbb{R} \rightarrow \mathbb{R}$. Specifically, we can show that if we let $\kappa(x) = a \cdot x$ be linear with gain $a\in \mathbb{R}_{\geq 0}$, then the set of vectors $\vec{u}\in \mathbb{U}$ satisfying \eqref{eq:vertices condition weaker} tends to the one satisfying \eqref{eq:vertices condition} for $a \rightarrow \infty$.
Reintroducing now the parameters $\vec{\vartheta}$, we can write \eqref{eq:vertices condition weaker} in matrix form, for each interval $(s_i,s_{i+1}),i\in \{1,2\}$, by stacking \eqref{eq:vertices condition weaker} over the index $k \in [[n_{h}]]$ and replacing the definition of $\gamma^{\varphi}$ to obtain
\begin{equation}\label{eq:mega compact represenation}
    E^i(\vec{\vartheta}) (\bar{A} \vec{z} + \bar{B} \vec{u} + \bar{\vec{p}}) \geq -\kappa (E^i(\vec{\vartheta})\vec{z} +\vec{g}^i(\vec{\vartheta})),
\end{equation}
where $E^i(\vec{\vartheta}) \in \mathbb{R}^{n_h \times (n+1)}$, $\vec{g}^i(\vec{\vartheta}) \in \mathbb{R}^{n_h}$, with $i\in \{1,2\}$, are given in \eqref{eq:normal vectors stack}. Note that $\kappa$ is applied element-wise when applied to a vector. 

In summary, by Prop. \eqref{prop:the piece of the puzzle that will make me the king}, if for each $(\vec{x},t) \in \mathbb{X}\times (s_i,s_{i+1})$ we are able to find $\vec{u}\in \mathbb{U}$ satisfying \eqref{eq:mega compact represenation}, then this implies that constraint \eqref{eq:condition single shot} can be satisfied for every $(\vec{x},t) \in \mathbb{X}\times (s_i,s_{i+1}), i\in \{1,2\}$, thus ensuring forward invariance of $\mathcal{B}^{\varphi}(t)$ under the input signal \eqref{eq:standard control input}. However, the sets $\mathbb{X}\times (s_{i}, s_{i+1}), i\in \{1,2\},$ contain a continuum of points for which we cannot verify numerically \eqref{eq:mega compact represenation}. We thus appeal to convexity, noting that each set $\mathbb{X}\times (s_{i}, s_{i+1}), i\in \{1,2\}$ is a polytope in the composite state and time domain (as Cartesian product of the polytopes $\mathbb{X}$ and $(s_{i}, s_{i+1})$), with vertices
\begin{equation}\label{eq:space time vertices}
V^{i}_{\mathbb{Z}} =  V_{\mathbb{X}} \times \{s_{i}, s_{i+1}\} \subset \mathbb{R}^{n+1}, \ i\in \{1,2\}
\end{equation}
such that $n^v_{\mathbb{Z}}:= |V^{i}_{\mathbb{Z}}| = 2|V_{\mathbb{X}}| = 2n^v_{\mathbb{X}}$, where we recall that $V_{\mathbb{X}}$ are the vertices of the set $\mathbb{X}$. Then every $\vec{z} \in \mathbb{X}\times (s_{i}, s_{i+1})$ is equivalently written as a convex combination $\vec{z} = \sum_{q=1}^{n^v_{\mathbb{Z}}} \lambda_q\vec{\eta}^i_q,$ with $\vec{\eta}^i_q \in V_{\mathbb{Z}}^{i},\; \lambda_q \geq 0,\; \forall q \in [[n^v_{\mathbb{Z}}]]$ and $\sum_{q=1}^{n^v_{\mathbb{Z}}}\lambda_q = 1
$. As usual, the index $i\in \{1,2\}$ denotes either of linear sections of $\gamma^{\varphi}$. The next example provides some intuition on the geometry of the problem.
\begin{example}
    Consider Figure \ref{fig:geometric intuition 1}, where the time-varying set $\mathcal{B}^{\varphi}(t)$ for a task $\varphi$ is shown in red over the space and time dimensions, such that $\mathbb{X}\subset \mathbb{R}^2$. By the time discontinuities $(s_i)_i^3 = (0,\alpha,\beta)$ induced by the function $\gamma^{\varphi}$, the set $\mathcal{B}^{\varphi}(t)$ is partitioned over the two intervals $\mathbb{X} \times (s_i,s_{i+1})$, namely $\mathbb{X} \times (0,\alpha)$ and $\mathbb{X} \times (\alpha,\beta)$. The vertices of the intervals $\mathbb{X}\times (s_i,s_{i+1}), i\in \{1,2\}$ are given by the Cartesian product of the vertices $V_{\mathbb{X}}$ (black dots) and the intervals $(s_i,s_{i+1})$ for a total of 8 vertices for each set $V^i_{\mathbb{Z}}$. \hfill $\square$
\end{example}

Leveraging this finite dimensional representation, Lemma \ref{lemma:numerical forward invariance} shows that the feasibility of \eqref{eq:mega compact represenation} over the continuum of points in $\mathbb{X} \times (s_i,s_{i+1}),\; i\in \{1,2\}$, can be verified by checking a set of linear inequalities over the set of vertices $V^{i}_{\mathbb{Z}}, i\in \{1,2\}$.
\begin{lemma}\label{lemma:numerical forward invariance}
Let an interval $(s_i,s_{i+1}), i\in \{1,2\}$ and let a linear class-$\mathcal{K}$ function $\kappa : \mathbb{R} \rightarrow \mathbb{R}$. If for each  $\vec{\eta}^i_q \in V_{\mathbb{Z}}^i$ there exists a vector $\vec{u}^i_q\in \mathbb{U}$, such that
\begin{equation}\label{eq: vertex constraint representation}
\begin{array}{l}
E^i(\vec{\vartheta}) (\bar{A} \vec{\eta}^i_q + \bar{B} \vec{u}^i_q + \bar{\vec{p}}) \geq -\kappa (E^i(\vec{\vartheta})\vec{\eta}^i_q + \vec{g}^i(\vec{\vartheta})) , 
\end{array}
\end{equation}
then $\mathcal{B}^{\varphi}(t)$ is forward invariant over $[0,\beta]$, under the control law \eqref{eq:standard control input}.
\end{lemma}
\begin{proof}
We want to prove that if \eqref{eq: vertex constraint representation} holds, then for each $\vec{x}_0 \in \mathbb{X}$, the input signal $\zeta_u: [0,\beta] \rightarrow \mathbb{U}$ given by \eqref{eq:standard control input} is well-defined, with $\zeta_x:  [0,\beta] \rightarrow \mathbb{X}$ being the solution to \eqref{eq:single agent dynamics} under $\zeta_u$. Namely, we prove that for every $(\vec{x},t) \in \mathbb{X} \times (s_i,s_{i+1})$ there exists $\vec{u}\in \mathbb{U}$ satisfying  \eqref{eq:mega compact represenation} from which the feasibility of the control law \eqref{eq:standard control input} is guaranteed by Prop. \ref{prop:the piece of the puzzle that will make me the king}. 

Independently of the index $i\in \{1,2\}$, let any point $\vec{z} \in \mathbb{X}\times (s_{i}, s_{i+1})$ such that $\vec{z} = \sum_{k=1}^{n_{\mathbb{V}_{\mathbb{Z}}}} \lambda_k\vec{\eta}^i_q$, where $\lambda_q\geq 0$ and $ \sum_{q=1}^{n_{\mathbb{V}_{\mathbb{Z}}}} \lambda_q=1$. Also let $\vec{u} =\sum_{q=1}^{n_{\mathbb{V}_{\mathbb{Z}}}} \lambda_q\vec{u}^i_q $ for the same coefficients $\lambda_q$ and note that $\vec{u}\in \mathbb{U}$ by convexity of $\mathbb{U}$. We can thus write (omitting $\vec{\vartheta}$ for brevity) 
$$
\begin{aligned}
&E^{i} (\bar{A} \vec{z} + \bar{B} \vec{u} + \bar{\vec{p}}) + \kappa (E^{i}\vec{z} + \vec{g}^i) =\\
&E^{i}\Bigl(\bar{A}\sum_{q=1}^{n^v_{\mathbb{Z}}}\lambda_q\vec{\eta}^i_q + \bar{B} \sum_{q=1}^{n^v_{\mathbb{Z}}} \lambda_q\vec{u}^i_q + \underbrace{\sum_{q=1}^{n^v_{\mathbb{Z}}} \lambda_q}_{1} \bar{\vec{p}} \Bigr)\\ 
&\qquad \qquad  + \kappa \Bigl(E^{i}\sum_{q=1}^{n^v_{\mathbb{Z}}}\lambda_q\vec{\eta}^i_q + \underbrace{\sum_{q=1}^{n^v_{\mathbb{Z}}} \lambda_q}_{1}  \vec{g}^i\Bigr) = \\
&E^{i} \sum_{q=1}^{n^v_{\mathbb{Z}}}\lambda_q\left(\bar{A}\vec{\eta}^i_q + \bar{B}\vec{u}^i_q + \bar{\vec{p}}\right) + \sum_{q=1}^{n^v_{\mathbb{Z}}}\lambda_q \kappa \left(E^{i}\vec{\eta}^i_q +\vec{g}^i\right) = \\
&\sum_{q=1}^{n^v_{\mathbb{Z}}} \lambda_q \left( E^{i} (\bar{A}\vec{\eta}^i_q + \bar{B}\vec{u}^i_q+ \bar{\vec{p}}) + \kappa (E^{i}\vec{\eta}^i_q + \vec{g}^i)\right) \underset{\eqref{eq: vertex constraint representation}}{\geq} 0. \\
\end{aligned}
$$    
Thus for every  $(\vec{x},t) \in \mathbb{X} \times (s_i,s_{i+1})$ the input $\vec{u} =\sum_{q=1}^{n_{\mathbb{V}_{\mathbb{Z}}}} \lambda_q\vec{u}^i_q \in \mathbb{U}$ satisfies \eqref{eq:mega compact represenation} and thus, by Prop. \ref{prop:the piece of the puzzle that will make me the king}, also satisfies constraint \eqref{eq:condition single shot}, such that \eqref{eq:standard control input} has at least one solution for each $(\vec{x},t)\in \mathbb{X}\times (s_i,s_{i+1}),i\in\{1,2\}$ and it is thus feasible over the intervals $(s_i,s_{i+1})$. We conclude noting that for every initial condition $\vec{x}_0\in \mathcal{B}^{\varphi}(0)$, if we apply $\zeta_u$ in \eqref{eq:standard control input}, then \eqref{eq:forward invariance varphi} is satisfied, proving, by Thm. \ref{thm:forward invariance non smooth theorem}, that $\mathcal{B}^{\varphi}(t)$ is forward invariant over the interval $[0,\beta] = [s_1,s_{3}]$.  
\end{proof}

In conclusion, the inequalities in \eqref{eq: vertex constraint representation}, are a set of linear constraints in the parameters $\vec{\vartheta}$ and the variables $\{\vec{u}^i_{q}\}_{q=1}^{n^i_{\mathbb{Z}}}$, that we can employ to enforce forward invariance of $\mathcal{B}^{\varphi}(t)$. We next generalize this idea to the conjunction set.

\begin{figure}[b]
    \centering
    \begin{subfigure}{0.49\linewidth}
        \centering
        =\includegraphics[width=\linewidth]{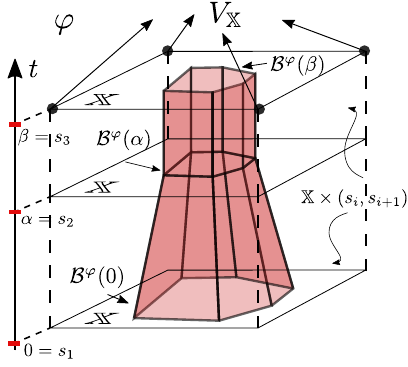}
        \caption{}
        \label{fig:geometric intuition 1}
    \end{subfigure}
    \hfill
    \begin{subfigure}{0.49\linewidth}
        \centering
        \includegraphics[width=\linewidth]{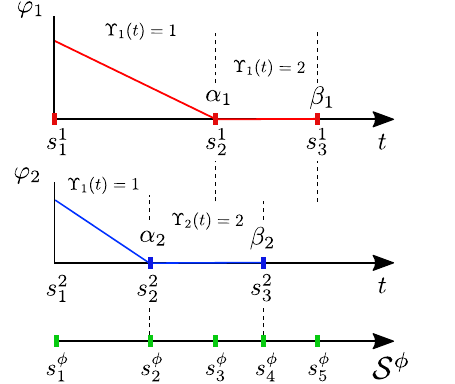}
        \caption{}
        \label{fig:geometric intuition 2}
    \end{subfigure}
    \caption{(a) Set $\mathcal{B}^{\varphi}(t)$ in the space-time dimension. Black dots represent the vertices $V^1_{\mathbb{Z}}$ and $V^2_{\mathbb{Z}}$ for the set $\mathbb{X}\times (0,\alpha)$ and $\mathbb{X}\times (\alpha,\beta)$, respectively. (b) Switching time sequence for two tasks $\varphi_l$ and $l\in \{1,2\}$. For each interval $(s_{j}^{\phi},s_{j+1}^{\phi})$ (green) there exists  index $i$ such that $(s^l_i,s^l_{i+1}) \supseteq (s_{j}^{\phi},s_{j+1}^{\phi})\; \forall l$.}
    \label{fig:geometric intuition}
\end{figure}

\subsection{Enforcing forward invariance : the conjunction case}\label{sec:forward invariance conjunciton}

We can now generalize the results of the previous section to guarantee forward invariance of the conjunction set $\mathcal{B}^{\phi}(t)$, as per \eqref{eq:conjunction set}, by applying Thm. \ref{thm:forward invariance non smooth theorem} to the barrier $\mathfrak{b}^{\phi}= \min_{l \in \mathcal{I}^{\phi}(t)} \{\mathfrak{b}_l^{\varphi}(\vec{x},t)\},$ as in \eqref{eq:super min barrier}, where each $\mathfrak{b}_l^{\varphi}(\vec{x},t)$ has a switching sequence $(s_i^l)_{i=1}^{3}=(0,\alpha_l,\beta_l)$. Namely, let the set
\begin{equation}\label{eq:parameters list times}
\mathcal{S}^{\phi}=\{0\} \cup \left(\bigcup_{l\in [[n_{\phi}]]}\{\alpha_l, \beta_l\} \right)= \left(\bigcup_{l\in [[n_{\phi}]]}(s_i^l)_{i=1}^{3} \right),
\end{equation}
gathering all the switching times, with $(s^{\phi}_j)_{j=1}^{2n_{\phi} +1}$ being the sequence obtained by ordering the elements in $\mathcal{S}^{\phi}$ (see Fig. \ref{fig:geometric intuition 2}). Recalling that the switch map $\mathcal{I}^{\phi}(t)$, as in  \eqref{eq:switch map}, is constant over the interval $[\beta_{l},\beta_{l+1})$, then $\mathcal{I}^{\phi}(t)$ is also constant over the intervals $[s_{j}^{\phi},s_{j+1}^{\phi} )$, since each $\beta_l$ is part of the sequence $(s^{\phi}_j)_{j=1}^{2n_{\phi}+1}$. Moreover, for each index $l \in \mathcal{I}^{\phi}(t)$ that is active over $(s_{j}^{\phi},s_{j+1}^{\phi})$, we have that either $(s_{j}^{\phi},s_{j+1}^{\phi}) \subseteq (s_1^l,s_2^l)$ or $ (s_{j}^{\phi},s_{j+1}^{\phi}) \subseteq (s_2^l,s_3^l)$, meaning that also the map $\Upsilon_l(t)$ is constant over $(s_{j}^{\phi},s_{j+1}^{\phi})$ (see Fig. \ref{fig:geometric intuition 2}). We can thus write 
\begin{equation}\label{eq:Constance relation}
\begin{aligned}
    \mathcal{I}^{\phi}(t) &= \mathcal{I}^{\phi}(s_j^{\phi}), \; \forall t\in (s_j^{\phi},s_{j+1}^{\phi}),\\
    \Upsilon_l(t)         &=  \Upsilon_l(s_j^{\phi}), \; \forall t\in (s_j^{\phi},s_{j+1}^{\phi}), \forall l\in [[n_{\phi}]]\\
\end{aligned}
\end{equation}
Similarly to the single task case, we need to find a control input such that the condition    
\begin{equation}\label{eq:forward invariance conjuntion}
\min \mathcal{L}_{\mathfrak{b}^{\phi}}(\zeta_x(t),t,\zeta_u(t)) \geq - \kappa(\mathfrak{b}^{\phi}(\zeta_x(t),t)), \; \forall t\in (s_j^{\phi},s_{j+1}^{\phi}), 
\end{equation}
is satisfied from every $\vec{x}_0 \in \mathcal{B}^{\phi}(0)$, from which forward invariance can be proved leveraging Thm. \ref{thm:forward invariance non smooth theorem}. With this goal, we select the control input given by 
\begin{subequations}\label{eq:standard control input conjunction}
\begin{gather}
    \zeta_u(t) = \text{argmin}_{\vec{u} \in \mathbb{U}} \; \|\vec{u}\|  \quad \text{s.t.} \\ 
     \begin{array}{l}
      \hspace{-0.638cm}\min \mathcal{L}^i_{\mathfrak{b}_l^{\varphi}}(\vec{x},t,\vec{u}) \geq - \kappa(\mathfrak{b}^{\varphi}(\vec{x},t)), i= \Upsilon_l(t), \forall l\in \mathcal{I}^{\phi}(t) \end{array}\label{eq:condition single shot conjunction} \\
     \vec{x} = \zeta_{x}(t), 
    \end{gather}
\end{subequations}
for all $t\in (s_j^{\phi},s_{j+1}^{\phi}), j\in [[2n_{\phi}]]$, while $\zeta_u(s^{\phi}_j)\; , \forall j\in [[2n_{\phi}]]$, can potentially take any bounded value, since this happens on a set of measure zero. Note that \eqref{eq:standard control input conjunction} is similar to \eqref{eq:standard control input} with the exception that now \eqref{eq:condition single shot conjunction} must be satisfied for each task index $l\in \mathcal{I}^{\phi}(t)$ among the tasks active at time $t$. The following result shows that if  \eqref{eq:standard control input conjunction} is feasible for each $t\in (s_j^{\phi},s_{j+1}^{\phi}), j\in [[2n_{\phi}]]$, then it satisfies \eqref{eq:forward invariance conjuntion}.
\begin{proposition}\label{prop:min lie derivative prop}
    Let an interval $(s_j^{\phi},s_{j+1}^{\phi}),\;j\in [[2n_{\phi}]]$. For any $(\vec{x},t) \in \mathbb{X} \times (s^{\phi}_{j},s^{\phi}_{j+1})$, if there exists $\vec{u} \in \mathbb{U}$ such that
     \begin{equation}\label{eq:lie derivative constraint in parallel}
     \begin{gathered}
      \min \mathcal{L}^i_{\mathfrak{b}^{\varphi}_l}(\vec{x},\vec{u},t) \geq -\kappa(\mathfrak{b}^{\varphi}_l(\vec{x},t)), \; i=\Upsilon_l(t),\; \forall l\in \mathcal{I}^{\phi}(t),
        \end{gathered}
    \end{equation}
    then $\min \mathcal{L}_{\mathfrak{b}^{\phi}}(\vec{x},\vec{u},t)  \geq -\kappa(\mathfrak{b}^{\phi}(\vec{x},t))$.
\end{proposition}
\begin{proof}
    Given in Appendix \ref{proof:min lie derivative prop}.
\end{proof}

Thus Prop. \ref{prop:min lie derivative prop} confirms that when the controller \eqref{eq:standard control input conjunction} is implemented over each interval $(s_j^{\phi},s_{j+1}^{\phi})$, then the constraint \eqref{eq:condition single shot conjunction} is such that condition \eqref{eq:forward invariance conjuntion} is satisfied and forward invariance of $\mathcal{B}^{\phi}(t)$ follows by Thm. \ref{thm:forward invariance non smooth theorem}. It thus remains to verify the feasibility of \eqref{eq:standard control input conjunction}. In this regard, we note that \eqref{eq:lie derivative constraint in parallel} is equivalent to \eqref{eq:condition single shot}, but repeated for every $l\in \mathcal{I}^{\phi}(t)$. We can thus reuse the same result of Prop. \ref{prop:the piece of the puzzle that will make me the king} to transform \eqref{eq:lie derivative constraint in parallel} into a set of linear inequalities alike \eqref{eq:mega compact represenation} from which the feasibility of \eqref{eq:standard control input conjunction} for each interval $\mathbb{X} \times (s_j^{\phi}, s_{j+1}^{\phi})$ can be checked numerically  over the vertices
\begin{equation}\label{eq: conjunction vertices}
V_{\mathbb{Z}}^j = V_{\mathbb{X}} \times \{s_j^{\phi}, s_{j+1}^{\phi}\},\;  \forall j\in [[2n_{\phi}]].  
\end{equation}
The proof of the next result is omitted as it follows the same steps of the proof of Lemma \ref{lemma:numerical forward invariance}.
\begin{lemma}\label{lemma: final lemma because I am done with the phd}
Let a linear class-$\mathcal{K}$ function $\kappa : \mathbb{R} \rightarrow \mathbb{R}$  and an interval $(s_j^{\phi}, s_{j+1}^{\phi}),\, j\in [[2n_{\phi}]]$. If for every $\vec{\eta}^j_q \in V_{\mathbb{Z}}^j$, there exists $\vec{u}^j_q\in \mathbb{U}$, such that 
\begin{equation}\label{eq: vertex constraint representation for parallel lie derivative}
\begin{aligned}
E_l^i(\vec{\vartheta}_l) (\bar{A} \vec{\eta}^j_q + \bar{B} \vec{u}^j_q + \bar{\vec{p}}) \geq -\kappa (E_l^i(\vec{\vartheta}_l)\vec{\eta}^j_q  +\vec{g}_l^i(\vec{\vartheta}_l)) ,\\
i = \Upsilon(s^{\phi}_j),\; \forall l\in \mathcal{I}^{\phi}(s^{\phi}_j),
\end{aligned}
\end{equation}
then $\mathcal{B}^{\varphi}(t)$ is forward invariant over $[s_j^{\phi},s_{j+1}^{\phi}]$. Moreover, if \eqref{eq: vertex constraint representation for parallel lie derivative} holds for every interval $(s_j^{\phi}, s_{j+1}^{\phi}),\, j\in [[2n_{\phi}]]$, then $\mathcal{B}^{\varphi}(t)$ is forward invariant and the control law \eqref{eq:standard control input conjunction} is well defined over the interval $[0,\beta_{\phi}]$.
\end{lemma}

We thus arrived at the desired result as we found a set of linear constraints on the parameters $\vec{\vartheta}_l,\;\forall l\in [[n_{\phi}]]$, namely \eqref{eq: vertex constraint representation for parallel lie derivative}, that guarantees the forward invariance of the set $\mathcal{B}^{\phi}(t|\vec{\theta})$ and which we can use in the optimization program \eqref{eq:very general optimization program}.
\subsection{Final optimization program}
We can now write the full optimization program we use for the selection of the optimal parameters $\vec{\theta}^{\star} = [\vec{\vartheta}_1^{\star T}, \ldots \vec{\vartheta}_{n_{\phi}}^{\star T}]^T$ (and recall $\vec{\vartheta}_l = [r_l, \bar{\gamma}_l]^T$) such that $\mathcal{B}^{\phi}(t|\vec{\theta}^{\star})$ is viable, forward invariant and the robustness of satisfaction for the task $\phi = \land_{l=1}^{n_{\phi}} \varphi_l$ is maximized. Namely, for each interval $ (s_j^{\phi}, s^{\phi}_{j+1})$, with $j\in [[2n_{\phi}]]$, let the sets of vertices $V_{\mathbb{Z}}^j= \{\vec{\eta}^j_{q}\}_{q=1}^{n^v_{\mathbb{Z}}},\; \forall j\in [[2n_{\phi}]]$, the set of input variables $V_{\mathbb{U}}^{j} = \{ \vec{u}_q^{j}\}_{q=1}^{n^v_{\mathbb{Z}}},\; \forall j\in [[2n_{\phi}]]$ and the union of these $V_{\mathbb{U}} = \bigcup_{j=1}^{2n_{\phi}} V_{\mathbb{U}}^j$. Moreover let the set of state variables $\Xi = \{\vec{\xi}_l\}_{l=1}^{n^{\phi}} \in \mathbb{R}^n$ and an initial known state $\vec{x}_0\in \mathbb{X}$ from which we seek to start the system, (e.g. the starting point from which we start planning a trajectory satisfyion $\phi$ as we show in the next section). The optimal parameters $\vec{\theta}^{\star}$ for the set $\mathcal{B}^{\phi}(t|\vec{\theta})$ are found as 

\begin{subequations}\label{eq:the final optimization program}
\begin{align}
&\qquad \qquad  \min_{\vec{\theta},  V_{\mathbb{U}}, \Xi} \; \quad  - \sum_{l=1}^{n_{\phi}} r_l \qquad \text{s.t.}\\
& \vec{\vartheta}_l \in \Theta_l, \; \vec{\xi}_l\in \mathbb{X}, \qquad \forall l\in [[n_{\phi}]]\label{eq:parameters constraints}\\
&\vec{u}_q^{j} \in \mathbb{U},\; \qquad \forall q\in [[n^v_{\mathbb{Z}}]],\; \forall j\in [[2n_{\phi}]],\label{eq:input constraints}\\
&E^1_{l}(\vec{\vartheta})  \begin{bmatrix} \vec{x}_0 \\ 0\end{bmatrix} + \vec{g}_{l}^1(\vec{\vartheta}) \geq 0,\, \begin{array}{l}
\forall l\in [[n_{\phi}]] \end{array}, \label{eq:non emptiness initial}\\
&\begin{aligned}
&E^i_{\tilde{l}}(\vec{\vartheta}) \begin{bmatrix} \vec{\xi}_l \\ \beta_l\end{bmatrix} + \vec{g}_{\tilde{l}}^i(\vec{\vartheta}) \geq 0, \\
& \hspace{0.8cm} i = \Upsilon_{\tilde{l}}(\beta_l), \; \forall l \in [[n_{\phi}]], \; \forall \tilde{l}\in \mathcal{I}^{\phi}(\beta_{l-1}),
\end{aligned} \label{eq:non emptiness} \\
&\begin{aligned}
&E_l^i(\vec{\vartheta}_l)(\bar{A} \vec{\eta}^j_q + \bar{B} \vec{u}^j_q + \bar{\vec{p}}) \geq -\kappa (E_l^i(\vec{\vartheta})\vec{\eta}^j_q  +\vec{g}_l^i(\vec{\vartheta}_l)) ,  \\
& i = \Upsilon_l(s_j^{\phi}),\;  \forall j\in [[2n_{\phi}]],\; \forall q\in [[n_{\mathbb{Z}}^v]],\; \forall l\in \mathcal{I}^{\phi}(s^{\phi}_j),\\
\end{aligned}
\label{eq:global constraint on forward invariance} 
\end{align}
\end{subequations}
where constraint \eqref{eq:parameters constraints} enforces each parameters $\vec{\vartheta}_l$ and state variables $\vec{\xi}_l$ in  their respective sets $\Theta_l$ and $\mathbb{X}$, constraint \eqref{eq:input constraints} enforces each input vector $\vec{u}_q^{j}$ in the sets $V_{\mathbb{U}}^{j}$ to be within the input bounds, constraint \eqref{eq:non emptiness initial} enforces  $\vec{x}_0 \in \mathcal{B}^{\phi}(0|\vec{\theta})$, constraint \eqref{eq:non emptiness} enforces viability of the set $\mathcal{B}^{\phi}(t|\vec{\theta})$ as per Proposition \ref{prop:discrete viability prop} and finally \eqref{eq:global constraint on forward invariance} enforces the forward invariance properties of the set $\mathcal{B}^{\phi}(t|\vec{\theta})$ as per Lemma \ref{lemma: final lemma because I am done with the phd}. Note that \eqref{eq:the final optimization program} is a linear program which can be solved efficiently using standard off-the-shelf solvers for very high-dimensional problems.

We can give an exact characterization of the dimensionality of \eqref{eq:the final optimization program}, for which we count $2n_{\phi}$ parameters in the vector $\vec{\theta}\in \mathbb{R}^{2n_{\phi}}$,  $2n_{\phi}\cdot m \cdot n^{v}_{\mathbb{Z}} = 4n_{\phi}\cdot m \cdot n^{v}_{\mathbb{X}}$ control variables $\vec{u}_q^{j}$ corresponding to a number $2n^{v}_{\mathbb{X}}$ of control variables with dimension $m$ for each interval $(s_i^{\phi},s_{i+1}^{\phi})$, and $n_{\phi}\cdot n$ state variables in the set $\Xi$ applied to check viability of the set $\mathcal{B}^{\phi}(t|\vec{\theta})$. Assuming that $\mathbb{X}$ is an hypercube for which  $n^v_{\mathbb{X}} = 2^n$, a lower bound on the dimensionality of \eqref{eq:the final optimization program} is $4n_{\phi}(2^{n} \cdot m+ \frac{n}{4} + \frac{1}{2})$. We conclude this section by summarizing the obtained result deriving from Lemma \ref{lemma: final lemma because I am done with the phd} and Lemma \ref{lemma:conjunction}.

\begin{theorem}\label{thm:validity lemma}
    If \eqref{eq:the final optimization program} is feasible, with solution $\vec{\theta}^{\star}$, then $\mathcal{B}^{\varphi}(t|\vec{\theta}^{\star})$ is forward invariant over the interval $[0,\beta_{\phi}]$. Moreover, for every trajectory $\zeta_x:[0, \beta_{\phi}] \rightarrow \mathbb{X}$ such that  $\zeta_x \rhd_{[0,\beta_{\phi}]} \mathcal{B}^{\phi}(t|\vec{\theta}^{\star})$, then $(\zeta_x,0) \models_r \phi$ with robustness $r = \min_{l} \{r_l^{\star}\}$.
\end{theorem}

When considering the disjunction $\psi= \lor_k \phi_k$, an instance of the optimization problem \eqref{eq:the final optimization program} is solved in parallel for each $\phi_k$, and the task with the highest robustness is then selected for execution to satisfy the disjunction $\psi$.
% \begin{remark}
%     \gl{We want to highlight that the infeasibility of \eqref{eq:the final optimization program} does not imply the impossibility of finding a time-varying set that can encode the STL task $\phi$ and that it is forward invariant with respect to the dynamics in \eqref{eq:single agent dynamics}. Indeed, the optimization program in \eqref{eq:the final optimization program} is parametric in the selected linear $\kappa$ function and on the selected switching times from Rules \ref{rule:eventually rule}-\ref{rule: always eventually}. We leave as future work to study the dependence of problem \eqref{eq:the final optimization program} on these parameters.}
% \end{remark}

\section{Trajectory planning via RRT$^{\star}$}\label{sec:rrt}
\begin{figure}[b]
    \centering
    \includegraphics[width=0.5\textwidth]{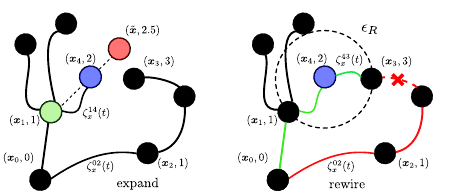}
    \caption{Example of an expansion phase of the $\text{RRT}^{*}$. (Left), the tree $\mathcal{T}$ is expanded toward the sampled node $(\tilde{\vec{x}},2.5)$ (red node) from the closest node $(\vec{x}_1,1)$ with time $1\leq 2.5$ (green node), until the new node $(\vec{x}_4,2)$ (blue node). (Right), the rewire phase, where the node $(\vec{x}_3,3)$ is rewired to the new node $(\vec{x}_4,2)$ to obtain a lower cost-to-go path (green) compared to the original path (red) in $\mathcal{T}$.}
    \label{fig:example of rrt algorithm}
\end{figure}
Given an STL task $\phi$, this section provides a detailed explanation of how Problem \ref{prob:problem 1} is approached leveraging the RRT$^{\star}$ planning algorithm. Loosely speaking, the approach taken consists on computing the time-varying set $\mathcal{B}^{\phi}(t)$, obtained as the solution to the parameters optimization problem \eqref{eq:the final optimization program} (dropping $\vec{\theta}$), and iteratively construct a tree of trajectories that evolves within the set $\mathcal{B}^{\phi}(t)$. It follows from the result of Lemma \ref{lemma:conjunction}, that every such trajectory robustly satisfies $\phi$. The reason why an RRT$^{\star}$ planning approach is relevant to solve Problem \ref{prob:problem 1} is that it allows to generate safe trajectories that satisfy a task $\phi$, by rejecting trajectories that can potentially hit obstacles in the environment.

\subsection{The algorithm}
RRT$^{\star}$ is a planning algorithm based on the iterative construction of a tree graph $\mathcal{T}(\mathcal{V}, \mathcal{E})$ of trajectories exploring the state space $\mathbb{X}$ (see Fig \ref{fig:example of rrt algorithm}). Namely, let a set $\mathcal{V} \subset \mathbb{Z}:= \mathbb{X}\times \mathbb{R}_{\geq 0}$ given by state-time pairs (or nodes) with $(\vec{x},t) \in \mathbb{Z}$, and let $\mathcal{E}$ be a set of \textit{edge} trajectories $\zeta_x^{ij}: [t_i,t_j] \rightarrow \mathbb{X}$ found as a  solution to \eqref{eq:single agent dynamics} and connecting node $(\vec{x}_i,t_i)$ to $(\vec{x}_j,t_j)$ such that $\zeta_x^{ij}(t_i) = \vec{x}_i$ and $\zeta_x^{ij}(t_j) = \vec{x}_j$. Let $\mathcal{O}_k,\; k\in [[n_o]]$, for some $n_o \geq 0$, represent a set of obstacles in $\mathbb{X}$ within which no valid trajectory should enter. Additionally, consider the root node $(\vec{x}_0,0)$ which defines the initial state from which the planning task is undertaken, and for each node index $i \in [[|\mathcal{V}|]]$, let the set-valued map
\begin{equation}\label{eq: path of nodes function}
\Pi(i) := \{ (\sigma_0 = 0,\sigma_1),(\sigma_1,\sigma_2), \ldots (\sigma_{l-1},\sigma_l= i)\},
\end{equation}
that returns the sequence of index pairs $(\sigma_k,\sigma_{k+1}) \in \mathbb{N} \times \mathbb{N}$ defining the unique path of length $l\geq 1$ from the root node $(\vec{x}_0,0)$ to $(\vec{x}_i,t_i) \in \mathcal{V}$. For example in Fig \ref{fig:example of rrt algorithm}-(left) we have that $\Pi(4) = \{(0,1),(1,4)\}$. Based the path defined by the map $\Pi$, the cost-to-go from the root node $(\vec{x}_0,0)$ to $(\vec{x}_i,t_i)$ is then computed by the function $\textbf{c2g}: \mathbb{N} \rightarrow \mathbb{R}_{\geq 0}$ as
\begin{equation}\label{eq:cost to go function}
\textbf{c2g}(i) := \sum_{(\sigma_k,\sigma_{k+1}) \in \Pi(i)} \int_{t_{\sigma_k}}^{t_{\sigma_{k+1}}} \|\dot{\zeta}_x^{\sigma_k,\sigma_{k+1}}(t)\|dt,
\end{equation}
which intuitively represents the total length of the trajectory from the root node $(\vec{x}_0,0)$ to $(\vec{x}_i,t_i)$ obtained by concatenating the edge trajectories $\zeta_x^{\sigma_k,\sigma_{k+1}}$. The objective is to construct a tree of trajectories with minimum cost, with respect to \eqref{eq:cost to go function}, that satisfy $\phi$. This is done in two main phases: an \textit{expansion} phase and a \textit{rewiring} phase, which we detail next.

\subsubsection{Initialization} At initialization (lines \ref{code:line:init1}-\ref{code:line:init2} in Alg. \ref{alg:rrt}) the root node $(\vec{x}_0,0)$ and a desired STL task $\psi = \lor_k \phi_k$ is provided, encoding a set of possible tasks to be achieved. Based on this information, the time-varying sets $\mathcal{B}_k^{\phi}(t)$ are computed by solving the parameter optimization \eqref{eq:the final optimization program} for each $\phi_k$. The task $\phi$ with the highest robustness is then selected for execution with the corresponding time-varying set $\mathcal{B}^{\phi}(t)$. Moreover, the tree $\mathcal{T}$ is initialized with $\mathcal{V} = \{(\vec{x}_0,0)\}$ and edge set $\mathcal{E} = \emptyset$.
\subsubsection{Expansion phase} After initialization, the main iteration of the algorithm starts (lines \ref{code:line:rand select}-\ref{code:line:steer} in Alg. \ref{alg:rrt}). During this phase a new sample $(\tilde{\vec{x}},\tilde{t}) \in \mathbb{Z}$ is obtained by first sampling a random time instant as $\tilde{t} \sim \textit{Unif}([0,t_{hr}(\phi)])$, where $\textit{Unif}(\cdot)$ represents a uniform probability distribution, and then a state $\tilde{\vec{x}}$ is sampled from the set $\mathcal{B}^{\phi}(\tilde{t})$ as $\tilde{\vec{x}} \sim \textit{Unif}(\mathcal{B}^{\phi}(\tilde{t}))$. We choose samples with maximum time $t_{hr}(\phi)$ as the task satisfaction is determined within this maximum time. Since $\mathcal{B}^{\phi}(t)$ is defined up to the time $\beta_{\phi} \leq t_{hr}(\phi)$ as per \eqref{eq:maximum satisfaction time}, we consider the convention $\mathcal{B}^{\phi}(t) = \mathbb{X}$ for all $t\in [\beta_{\phi}, t_{hr}(\phi)]$, which is without loss of generality since the satisfaction of $\phi$ is determined within the time $\beta_{\phi}$ as per Lemma \ref{lemma: final lemma because I am done with the phd}.

After sampling, the nearest neighbour $(\vec{x}_i,t_i) \in \mathcal{V}$ to $(\tilde{\vec{x}},\tilde{t})$ is found via the function $\textbf{past\_nn}: \mathbb{Z} \rightarrow  \mathcal{V}$, defined as 
\begin{equation}\label{eq:past nn}
\begin{aligned}
\textbf{past\_nn}(\tilde{\vec{x}},\tilde{t}) = \; &\text{argmin}_{(\vec{x},t) \in \mathcal{V}} \; \; \textbf{dist}((\tilde{\vec{x}},\tilde{t}),(\vec{x},t)),\\
&\qquad s.t. \quad t\leq \tilde{t}, 
\end{aligned}
\end{equation}
using the distance metric $\textbf{dist} : \mathbb{Z} \times \mathbb{Z} \rightarrow \mathbb{R}_{\geq0}$ given as
\begin{equation}
\textbf{dist}((\tilde{\vec{x}},\tilde{t}),(\vec{x},t)) = \| \tilde{\vec{x}} - \vec{x}\|_2 + |\tilde{t} -t|,
\end{equation}
which considered both the distance in terms of time and space. By \eqref{eq:past nn} the nearest neighbour $(\vec{x}_i,t_i)= \textbf{past\_nn}(\tilde{\vec{x}},\tilde{t})$ is such that $t_i \leq \tilde{t}$, to maintain the time consistency of the tree during the expansion. Indeed, after the nearest neighbour $(\vec{x}_i,t_i)$ is found, we apply the function  $\textbf{steer} : \mathcal{V} \times \mathbb{Z} \rightarrow \mathcal{E}$ which returns the trajectory $\zeta_x^{ij} = \textbf{steer}((\vec{x}_i,t_i), (\tilde{\vec{x}},\tilde{t}))$, obtained by solving
\begin{subequations}\label{eq:steer function}
    \begin{align}
    \min_{\zeta^{ij}_{u},\zeta^{ij}_{x}} \int_{t_i}^{t_i+\Delta} \|\zeta^{ij}_{x}(t)- \tilde{\vec{x}}\|^2_{Q} +  \|\zeta^{ij}_{u}(t)\|^2_{R} \; dt
    \end{align}
    \begin{align}
    \zeta^{ij}_{u}(t)&\in \mathbb{U},  &\forall t\in [t_i,t_j]  \label{eq:input and state bounds}  \\
    \dot{\zeta}^{ij}_{x}(t) &= A\zeta^{ij}_{x}(t) + B\zeta^{ij}_{u}(t) + \vec{p},\; a.e.\; &t\in [t_i,t_j]  \label{eq:dynamics bound}\\
    \zeta^{ij}_{x}(t) &\in \mathcal{B}^{\phi}(t),   &\forall t\in [t_i,t_j] \label{eq:barrier constraint}\\
    \zeta^{ij}_{x}(t_i) &= \vec{x}_i, 
\end{align}
\end{subequations}
where $Q\in \mathbb{R}^{n\times n}$ and $R \in \mathbb{R}^{m\times m}$ are two positive definite matrices, and $\Delta = \min\{(\tilde{t}-t_i), \delta_{max}\}$ with $\delta_{max} \geq 0$ being a maximum expansion interval. Intuitively, the optimal trajectory $\zeta^{ij}_{x}$ obtained by solving \eqref{eq:steer function} is such that it starts from $(\vec{x}_i,t_i)$ and it moves toward the node $(\tilde{\vec{x}},\tilde{t})$ for a time interval $\Delta$ until the node $(\vec{x}_j,t_j)$ with $t_j = t_i + \Delta$ and $\vec{x}_j = \zeta_x^{ij}(t_j)$, as depicted in  Fig. \ref{fig:example of rrt algorithm}-(left). Note constraint \eqref{eq:barrier constraint} is a set of time-varying linear constraints with the form \eqref{eq:polyhedral representation}, such that \eqref{eq:steer function} is a quadratic program (quadratic cost and linear constraints), which can be solved efficiently.

An important fact is that, thanks to the construction approach taken for $\mathcal{B}^{\phi}(t)$, the control input signal $\zeta^{ij}_{u} : [t_i,t_j] \rightarrow \mathbb{U}$ generated by applying \eqref{eq:standard control input conjunction} from the initial state $\vec{x}_i$ (with resulting state signal $\zeta^{ij}_{x}: [t_i,t_j] \rightarrow \mathbb{U}$) always represents a feasible solution to \eqref{eq:steer function} as per Lemma \ref{lemma: final lemma because I am done with the phd}, from which the following proposition is derived.

\begin{proposition}
For each $(\vec{x}_i, t_i)\in \mathcal{V}$ in  $\mathcal{T}$ it holds $\vec{x}_i \in \mathcal{B}^{\phi}(t_i)$ and a solution $\zeta_x^{ij}: [t_i,t_j] \rightarrow \mathbb{X}$ to \eqref{eq:steer function} always exists, from any $(\vec{x}_i, t_i)\in \mathcal{V}$, with $\zeta_{x}^{ij}(t) \rhd_{[t_i,t_j]} \mathcal{B}^{\phi}(t)$.
\end{proposition}
% \begin{proof}
%     Let initially the node $(\vec{x}_0,t_0) \in \mathcal{V}$ such that $\vec{x}_0 \in \mathcal{B}^{\phi}(t_0)$. For any sampled node $(\tilde{\vec{x}}, \tilde{t})$,  the node $(\vec{x}_0,t_0)$ will be the unique nearest neighbor in this case, for which we seek an optimal edge trajectory $\zeta_{x}^{ij} : [t_i,t_i+\Delta] \rightarrow \mathbb{X}$ satisfying \eqref{eq:steer function}. A feasible solution to \eqref{eq:steer function} is given by applying the control input \eqref{eq:standard control input conjunction} over the interval $[0,\Delta]$ starting from $\vec{x}_0$. Indeed, we know by Lemma \ref{lemma: final lemma because I am done with the phd} that \eqref{eq:standard control input conjunction} is always well defined, respects the control bound $\mathbb{U}$ and guarantees $\zeta_{x}^{ij}(t) \in \mathcal{B}^{\phi}(t)\; \forall t\in [0,\Delta]$. Noting that $t_j = \Delta$, we thus conclude that $\vec{x}_j = \zeta_{x}^{ij}(t_j) \in \mathcal{B}^{\phi}(t_j)$, which is added to the set $\mathcal{V}$. The argument is repeated recursively, for each new sampled node $(\tilde{\vec{x}}, \tilde{t})$ and nearest neighbour $(\vec{x}_i,t_i) \in \mathcal{V}$, concluding that $\vec{x}_i \in \mathcal{B}^{\phi}(t_i)$, for all $(\vec{x}_i,t_i) \in \mathcal{V}$.
% \end{proof}

Once the edge trajectory $\zeta^{ij}_{x}$ is found as a solution to \eqref{eq:steer function}, then the trajectory is checked for collision against the set of obstacles $\mathcal{O}_k, \; \forall k \in [[n_{o}]]$ in the workspace, to ensure safety. If a collision is found i.e. there exists $t_c \in [t_i,t_j]$ such that $\zeta_x^{ij}(t_c)\in \mathcal{O}_k$, for some $k\in [[n_{o}]]$, then the trajectory is rejected, and a new sample is drawn, otherwise the trajectory $\zeta_x^{ij}$ is added to $\mathcal{E}$ and the new node $(\vec{x}_j,t_j)$ is added to $\mathcal{V}$ thus leading the the rewiring phase. 
\subsubsection{Rewiring phase} After the expansion phase has taken place, a set $\mathcal{R} = \{(\vec{x}_r,t_r)\}_r \subseteq \mathcal{V}$ of nodes in $\mathcal{T}$ is retrieved in the neighborhood of the newly connected node $(\vec{x}_j,t_j)$ by the function $\textbf{future\_nns} : \mathcal{V} \rightarrow 2^{\mathcal{V}}$ given by 
\begin{equation}
\begin{aligned}
&\textbf{future\_nns}((\vec{x}_j,t_j)) = \{ (\vec{x}_r,t_r) \in \mathcal{V} \mid \\
&\hspace{1cm}\textbf{dist}((\vec{x}_r,t_r),(\vec{x}_j,t_j))\leq \epsilon_R \land\;  t_r \geq t_j\}.
\end{aligned}
\end{equation}
with $\epsilon_{R}>0$ being a given \textit{rewiring radius}. Intuitively, the set $\mathcal{R} = \textbf{future\_nns}((\vec{x}_j,t_j))$ is the set of nodes within a ball of radius $\epsilon_{R} > 0$ from $(\vec{x}_j,t_j)$ and with $t_{r} \geq t_j$ (see right panel in Fig. \ref{fig:example of rrt algorithm}) and for which we want to reduce the current cost-to-go in the tree $\mathcal{T}$. The constraint $t_{r} \geq t_j$ serves to guarantee the time consistency in the expansion of the tree during the rewiring. Indeed, for each node $(\vec{x}_r,t_r) \in \mathcal{R}$ a trajectory $\zeta_{x}^{jr}$ is computed from $(\vec{x}_j,t_j)$ to $(\vec{x}_r,t_r)$ by the function $\zeta_{x}^{jr} = \textbf{bridge}((\vec{x}_j,t_j),(\vec{x}_r,t_r))$, which is given by the solution of the following optimization program
\begin{subequations}\label{eq:bridge function}
    \begin{align}
    \min_{\zeta^{jr}_{u},\zeta^{jr}_{x}} \int_{t_j}^{t_r} \|\zeta^{jr}_{x}(t)\|_{Q} +  \|\zeta^{jr}_{u}(t)\|_{R} \; dt
    \end{align}
    \begin{align}
    \zeta^{jr}_{u}(t)&\in \mathbb{U},  &\forall t\in [t_j,t_r]   \\
    \dot{\zeta}^{jr}_{x}(t) &= A\zeta^{jr}_{x}(t) + B\zeta^{jr}_{u}(t) + \vec{p}, \; a.e.\; &t\in [t_j,t_r] \\
    \zeta^{jr}_{x}(t) &\in \mathcal{B}^{\phi}(t),   &\forall t\in [t_j,t_r] \\
    \zeta^{jr}_{x}(t_i) &= \vec{x}_j, \quad  \zeta^{jr}_{x}(t_r) = \vec{x}_r\label{eq: boundary value constraint}
\end{align}
\end{subequations}
which is similar to \eqref{eq:steer function}, with the additional constraint \eqref{eq: boundary value constraint} that enforces the terminal value constraint $\zeta^{ir}_{x}(t_r) = \vec{x}_r$. In principle, \eqref{eq:bridge function} is computationally harder compared to \eqref{eq:steer function}, due to the double boundary value constraint. Moreover, different from \eqref{eq:steer function}, the feasibility of \eqref{eq: boundary value constraint} can not be guaranteed due to the boundary value constraint. At this point, a trajectory is considered for rewiring if 1) a solution to \eqref{eq:bridge function} exists, 2) such solution is not in collision and 3) the cost-to-go for node $(\vec{x}_r,t_r)$ is reduced by the rewiring through node $(\vec{x}_j,t_j)$ as 
$$
\underbrace{\textbf{c2g}(j) +  \int_{t_j}^{t_r} \|\zeta^{jr}_x(t)\|dt}_{\text{potential new cost for node $(\vec{x}_r,t_r)$}} \leq \textbf{c2g}(r) 
$$
then the new trajectory $\zeta^{jr}_x$ is added to the tree and the old trajectory  $\zeta^{pr}_x$ from the parent node $(\vec{x}_p,t_p)$ of node $(\vec{x}_j,t_j)$ in the tree $\mathcal{T}$ is removed. In any other case, the current rewiring failed, and the algorithm continues.

The rewiring phase terminates when all possible rewiring attempts have been made, and the whole iteration is restarted.
\subsubsection{Termination} When the algorithm terminates, we are left with the set of nodes $\mathcal{V}$ and edge trajectories in $\mathcal{E}$ where each trajectory $\zeta^{ij}_x \in \mathcal{E}$ is such that $\zeta^{ij}_x(t) \in \mathcal{B}^{\phi}(t),\; \forall t\in [t_i,t_j]$. Thus, for every node $(\vec{x}_s,t_s)\in \mathcal{V}$, such that $t_s \geq t_{hr}(\phi)$, recall the definition of the map $\Pi(s)$ defining the indices of the path of nodes connecting the initial node $(\vec{x}_0,0)$ to $(\vec{x}_s,t_s)$. A satisfying trajectory $\zeta_{x}:[0,t_s] \rightarrow \mathbb{X}$ robustly satisfying $\phi$ is obtained by concatenating the trajectories along the edges in $\Pi(s)$ as shown next. 
\begin{theorem}\label{thm:the last theorem}
    For any node  $(\vec{x}_s,t_s)\in \mathcal{V}$ obtained by Alg. \ref{alg:rrt} such that $t_s \geq t_{hr}(\phi)$, let the signal $\zeta_x: [0,t_s] \rightarrow \mathbb{X}$ defined as the concatenation  
   \begin{equation}\label{eq:concatenation}
    \zeta_x(t) = \zeta_x^{\sigma_{k},\sigma_{k+1}}(t),\;  \forall t\in [t_{\sigma_{k}},t_{\sigma_{k+1}}], \;\forall (\sigma_{k},\sigma_{k+1}) \in \Pi(s), 
    \end{equation}
    such that $t_{\sigma_0}=0$ and $t_{\sigma_{l}}=t_s$ with $l = |\Pi(s)|$. Then $(\zeta_x,0) \models_{r} \phi$ and $\zeta_x$ respects the dynamics \eqref{eq:single agent dynamics}.
\end{theorem}
\begin{proof}
    The result follows directly by the fact that each trajectory $\zeta_x^{\sigma_{k},\sigma_{k+1}} \in \mathcal{E}$ is dynamically feasible and such that $\zeta_x^{\sigma_{k},\sigma_{k+1}}(t)\in\mathcal{B}^{\phi}(t),\; \forall t\in [t_{\sigma_{k}},t_{\sigma_{k+1}}]$, as it is the solution to either \eqref{eq:steer function} or \eqref{eq:bridge function}. Since $\zeta_x: [0,t_s] \rightarrow \mathbb{X}$ in \eqref{eq:concatenation} is the concatenation of the sequence of edge trajectries $\zeta_x^{\sigma_{k},\sigma_{k+1}}$, then $\zeta_x \in \mathcal{B}^{\phi}(t),\; \forall t\in [0,t_{s}]$ and thus  $(\zeta_{x} ,0) \models_r \phi$ with robustness $r>0$ as per Lemma \ref{lemma:conjunction}.
\end{proof}

Thus, our solution to Problem \ref{prob:problem 1} is given by selecting the node $(\vec{x}^{\star},t^{\star})\in \mathcal{V}$ and the corresponding trajectory $\zeta^{\star}: [0,t_s] \rightarrow \mathbb{X}$, with $t^{\star} \geq t_{hr}(\phi)$ and minimum cost-to-go (Alg. \ref{alg:rrt} line \ref{code:line:last line}). By Thm. \ref{thm:the last theorem} such trajectory robustly satisfies $\phi$ and thus the original disjunction $\psi$.
\begin{algorithm}
\caption{Modified RRT$^{\star}$ Algorithm}
\begin{algorithmic}[1]
\State $\textbf{Input} \quad \psi = \lor_k \phi_k \; \text{and} \; \vec{x}_0 \in \mathbb{X}$
\State Compute $\mathcal{B}^{\phi}_k(t)$ for each $\phi_k$ by \eqref{eq:the final optimization program}\label{code:line:init1}
\State Select $\mathcal{B}^{\phi}(t)$ with maximum robustness
\State Initialize $\mathcal{V} \gets \{(\vec{x}_{0},0)\}$, $\mathcal{E} \gets \emptyset$ with $\vec{x}_0 \in \mathcal{B}^{\phi}(0)$\label{code:line:init2}
\For{$i = 1$ to $n_{max}$}\label{code:line:for loop start}
    \State Sample $\tilde{t} \sim \textit{Unif}([0,t_{hr}(\phi)])$ and $\tilde{\vec{x}} \sim \textit{Unif}(\mathcal{B}^{\phi}(\tilde{t}))$  \label{code:line:rand select}
    \State $(\vec{x}_{i},t_{i}) \gets \textbf{past\_nn}((\tilde{\vec{x}},\tilde{t}))$
    \State $\vec{\zeta}_x^{ij}(t) \gets \textbf{steer}((\vec{x}_{i}, t_{i}),(\tilde{\vec{x}},\tilde{t}))$ \label{code:line:steer}
    \If{ $\vec{\zeta}_x^{ij}(t)$ is \textbf{not} in collision}
        \State $\mathcal{V} \gets \mathcal{V} \cup \{(\vec{x}_{j},t_{j})\}$ 
        \State $\mathcal{E} \gets \mathcal{E} \cup \{\zeta_x^{ij}(t)\}$
        \State $\mathcal{R} = \{(\vec{x}_r,t_{r})\}_r \gets \textbf{future\_nns}((\vec{x}_{j},t_j))$
        \ForAll{$(\vec{x}_r,t_{r})\in \mathcal{R}$}
            \State $\mathcal{E} \gets \textbf{rewire}((\vec{x}_{j},t_j),(\vec{x}_{r},t_r)) $
        \EndFor
    \Else{}
        \State continue
    \EndIf
\EndFor
\State $\zeta^{\star} \gets \textbf{get\_trajectory}(\mathcal{T})$\label{code:line:last line}
\State \Return $\zeta^{\star}$ \label{code:line:for loop end}
\end{algorithmic}
\label{alg:rrt}
\end{algorithm}

\vspace{-0.6cm}
\begin{algorithm}
\caption{\textbf{rewire}}
\begin{algorithmic}[1]
\State \textbf{Function} \textbf{rewire}{($(\vec{x}_{j},t_j), (\vec{x}_{r},t_r))$}
\State $\zeta_x^{jr}(t) \gets \textbf{bridge}((\vec{x}_{j}, t_{j}),(\vec{x}_{r}, t_{r}))$
\If{$\textbf{ctg}((\vec{x}_{j}, t_{j})) + \int_{t_j}^{t_r} \|\dot{\zeta}_{x}^{jr}(t)\| dt<\textbf{ctg}(r)$} \label{code:line:cost check}
        \If{$\zeta^{jr}_x(t)$ is \textbf{not} in collision} \label{code:line:start rewire}
            \State Get $(\vec{x}_{p},t_{p})$, parent of $(\vec{x}_{r},t_{r}))$ in $\mathcal{T}$
            \State $\mathcal{E} \gets (\mathcal{E}\setminus \{\zeta^{pr}_x\}) \cup \{\zeta_x^{jr}\}$ \label{code:line:end rewire}
        \EndIf
 \EndIf
 \Return $\mathcal{E}$
\end{algorithmic}
\label{alg:rewire}
\end{algorithm}

\section{Simulations}\label{sec:sim}
\begin{figure*}[t!]
    \centering

    % Subfigure 1
    \begin{subfigure}[b]{0.3\textwidth}
        \centering
        \includegraphics[width=\linewidth]{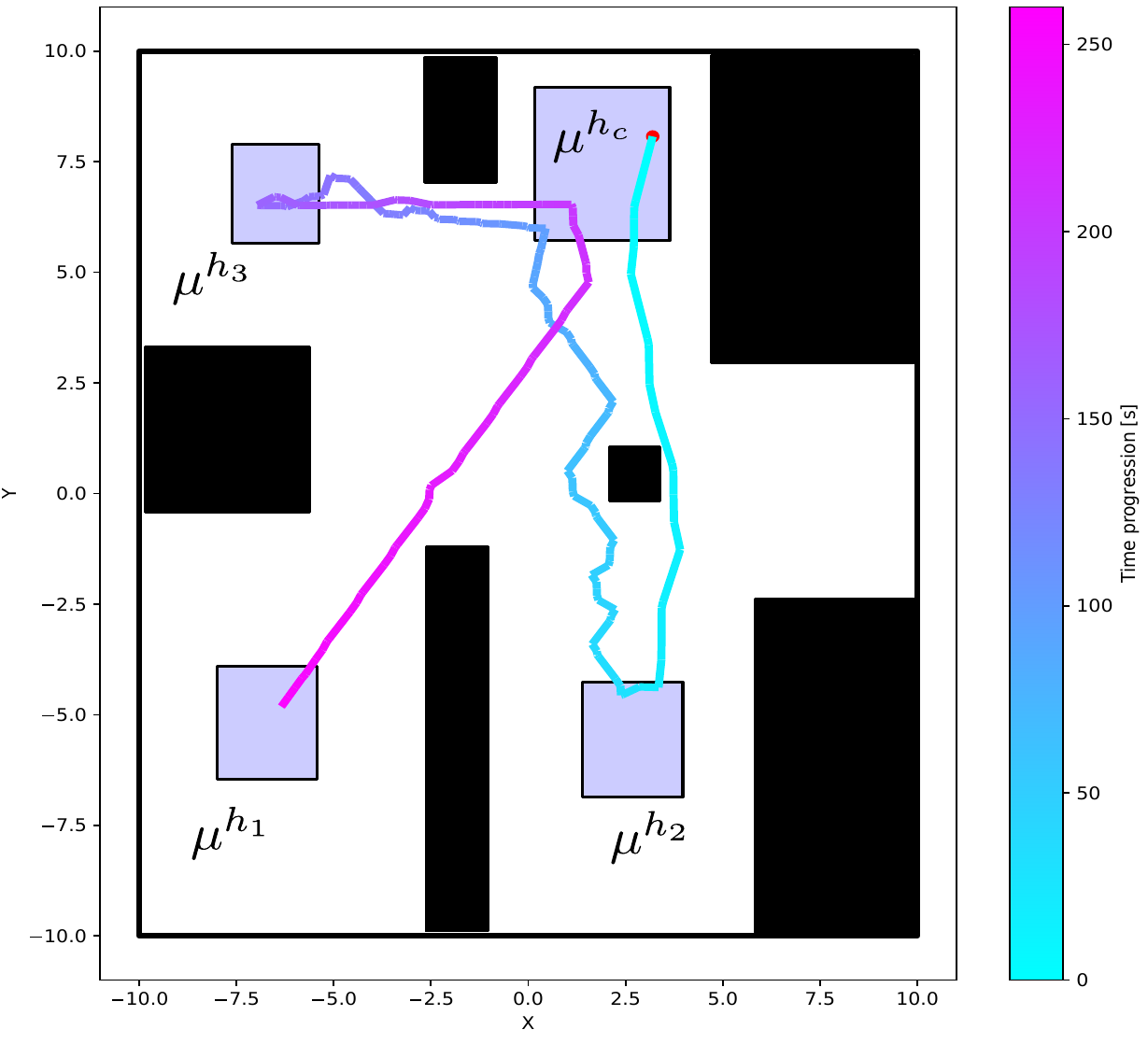}
        \caption{}
        \label{fig:room service}
    \end{subfigure}
    \hspace{0.3cm}
        \begin{subfigure}[b]{0.3\textwidth}
        \centering
        \scriptsize
        \begin{tabular}{lll|ll}
        \toprule[2pt]\toprule
         \multicolumn{5}{c}{RRT Planning} \\\midrule
        & \multicolumn{2}{l}{Case (1)} & \multicolumn{2}{l}{Case (2)} \\\midrule
        robustness        &  \multicolumn{2}{l}{0.12} & \multicolumn{2}{l}{1.58} \\
        task opt. (s) &  \multicolumn{2}{l}{0.019$^{*}$}        & \multicolumn{2}{l}{0.024} \\ 
        \midrule[2pt]
        planning          & First & Best  & First & Best\\
        \midrule
        avg. cost         & 65.76 & 65.33 & 540.23 & 538.92 \\
        std. cost         & 0.81  & 0.71  & 2.31 & 2.33 \\
        avg. time (s)     & 1.39  & 3.0   & 3.77 & 6.06 \\
        std. time (s)     & 0.51  & 1.05  & 0.59 & 1.19 \\
        \bottomrule
        \bottomrule
         \multicolumn{5}{l}{\text{$*$ Sum of solver times of \eqref{eq:the final optimization program} for tasks $\phi_1$ and $\phi_2$}}
        \end{tabular}
        \caption{}
        \label{tab:simulations summary}
    \end{subfigure}
    \hspace{0.3cm}
    \begin{subfigure}[b]{0.3\textwidth}
        \centering
        \includegraphics[width=1.1\linewidth]{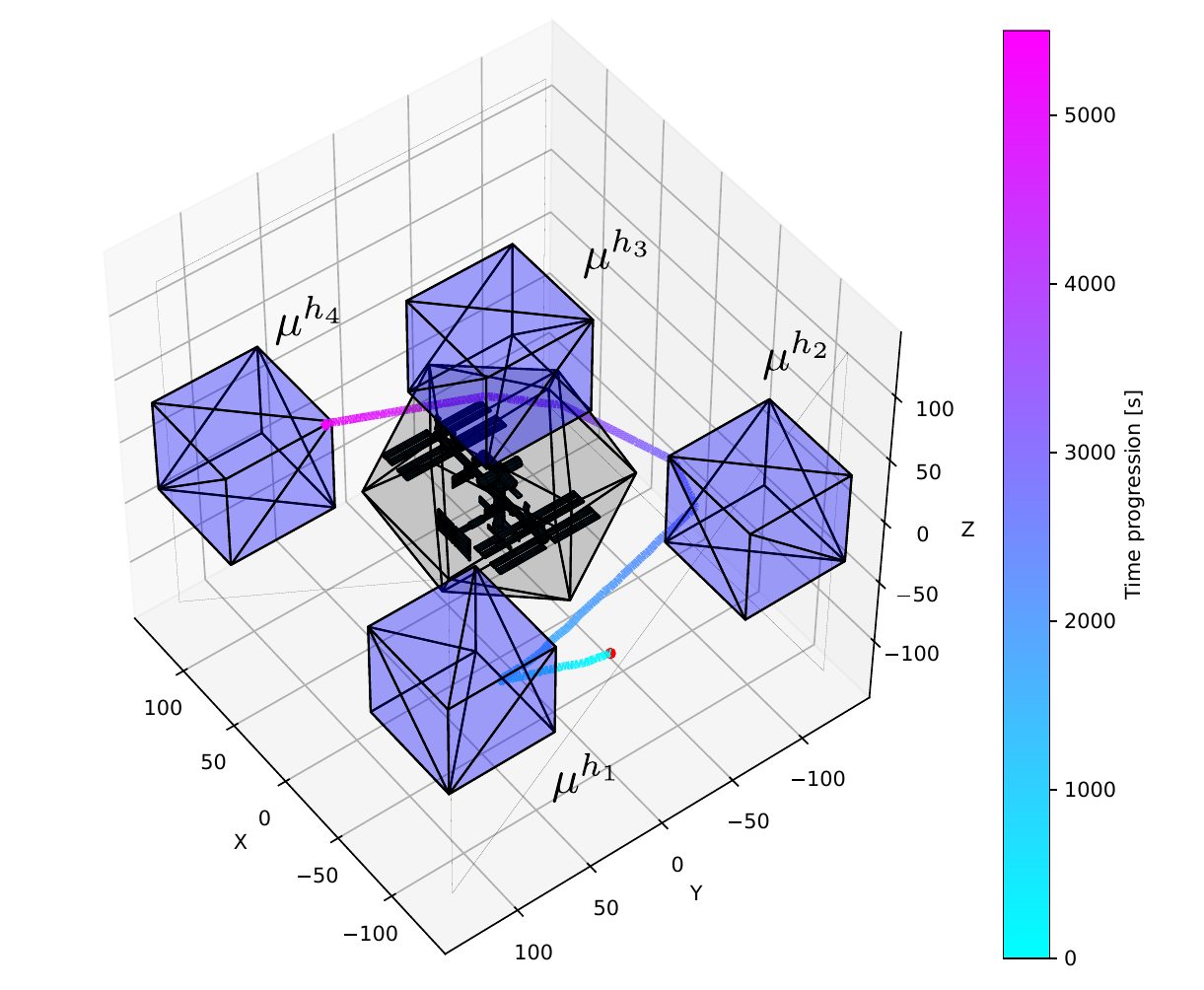}
        \caption{}
        \label{fig:iss inspection}
    \end{subfigure}

    \caption{Trajectories planned by Alg. \ref{alg:rrt} for the room servicing scenario (a) and ISS inspection (c). The color gradient denotes time progression. On panel (b), we show the performance summary table. First, the task optimization time required to compute the set $\mathcal{B}^{\phi}(t)$ by solving \eqref{eq:the final optimization program} (Alg. \ref {alg:rrt} line \ref{code:line:init1} ) is reported with the obtained robustness. Then, the average and standard deviation of the cost/time for the first/best solution found by our RRT$^{\star}$ algorithm for the two case studies (the time is computed based on the RRT$^{\star}$ iteration Alg. \ref{alg:rrt} from line \ref{code:line:for loop start} to\ref{code:line:for loop start}) is reported based on a batch of 100 simulations. Simulations where run on Intel Core i7-1265U with 32GB RAM.}
    \label{fig:all-scenarios}
\end{figure*}

We present two use cases of our proposed RRT algorithm. One involves servicing some interest points while ensuring revisit of a charging station, and one involves inspecting the International Space Station (ISS) with a deputy spacecraft over a long-horizon mission. The open source software \texttt{cvxpy} \cite{diamond2016cvxpy} was employed to solve the linear program \eqref{eq:the final optimization program} for each STL task and to transcribe the continuous time optimization programs \eqref{eq:steer function} and \eqref{eq:bridge function} necessary to run Alg. \ref{alg:rrt}.
\subsubsection*{Case study (1)}
Consider a robot with state $\vec{x} = \vec{p}$, where $\vec{p}$ is the position of the robot moving in the workspace $\mathbb{X} = \{ \vec{p} \in \mathbb{R}^2 \mid \|\vec{p}\|_{\infty} \leq 10 \}$ (see Fig. \ref{fig:room service}) and input bound $\mathbb{U} = \{ \vec{u} \in \mathbb{R}^2 \mid \|\vec{u}\|_{\infty} \leq 5 \}$ such that
\begin{equation}\label{eq:single integrator dyn}
\dot{\vec{p}} = \begin{bmatrix} -0.049 & -0.029 \\ -0.071 & -0.049 \end{bmatrix} \vec{p} +  \vec{u}
 \end{equation}
 The dynamics \eqref{eq:single integrator dyn} represent a single integrator dynamics with a position-dependent velocity drift. We assign the robot the specification $\psi= \phi_1 \lor \phi_2$ where 
 $\phi_1 = F_{[20,25]} \mu^{h_1}  \land F_{[120,150]} \mu^{h_2} \land G_{[0,200]} F_{[0,100]} \mu^{h_c}  \land G_{[255,265]} \mu^{h_3}$ and $\phi_2 = F_{[20,25]} \mu^{h_2}  \land F_{[120,150]} \mu^{h_3} \land G_{[0,200]} F_{[0,100]} \mu^{h_c}  \land G_{[255,265]} \mu^{h_1}$ where the predicate functions $h_{i}, i\in \{1,2,3\}$ represent areas of interest, while $h_c$ represents the charging station which should be revisited  (see purple boxes in Fig. \ref{fig:room service}). The difference between $\phi_1$ and $\phi_2$ is only in the order in which the regions of interest are reached. The task $\phi_2$ was selected for the satisfaction of $\psi$ with maximum robustness $r = 0.12$, while the resulting robustness for $\phi_1$ was only $r=0.06$. The best trajectory (in terms of the cost in \eqref{eq:cost to go function}) obtained by our algorithm is shown in Fig. \ref{fig:room service}, with the color gradient representing the time progression along the trajectory. The algorithm was run for a total of 500 iterations, over 100 experiments with running time and trajectory cost summarized in Tab. \ref{tab:simulations summary}. Due to the sampling-based nature of the algorithm, the obtained trajectory is only piece-wise smooth. While we did not smooth the resulting trajectory, any smoothing algorithm could be applied to further improve the final result, e.g., using splines optimization \cite{splines}.

\subsubsection*{Case study (2)}
We consider a deputy spacecraft inspecting the International Space Station (ISS) by visiting a set of pre-selected observation regions (purple boxes in Fig. \ref{fig:iss inspection}). The deputy is modelled as a double integrator governed by the standard Clohessy–Wiltshire model \cite{clohessy1960terminal}. Namely, let the state $\vec{x} = [\vec{p}\; \vec{v}]^T \in \mathbb{R}^6$ represent the position and velocity of the deputy with dynamics $\dot{\vec{x}} = A\vec{x} + B \vec{u}$ as
\begin{equation}
A = \left[
\begin{smallmatrix*}
 & 0_{3\times 3} &  &  & I_{3 \times 3} &  \\
\hline
3n^2 & 0 & 0 & 0 & 2n & 0 \\
0 & 0 & 0 & -2n & 0 & 0 \\
0 & 0 & -n^2 & 0 & 0 & 0 \\
\end{smallmatrix*}
\right], B =
\begin{bmatrix}
0_{3\times 3}\\
I_{3}
\end{bmatrix}
\end{equation}
where $n=1.13e^{-3}$ for the ISS, and $\mathbb{X} := \{ \vec{x} \in \mathbb{R}^6 \mid \|\vec{p}\|_{\infty} \leq 120,\;\|\vec{v}\|_{\infty} \leq 0.5\}$ and $\mathbb{U} := \{ \vec{u}\in \mathbb{R}^3 \mid \|\vec{u}\|_\infty \leq 1.5 \}$. The deputy is required to visit four observation regions (purple boxes in  Fig. \ref{fig:iss inspection}) and hold the position for a period of $400s$ ($\approx 6.7$ min). We encode this task in the STL formalism as (expressing time intervals in minutes) $\psi = \phi = F_{[16.7,23.3]}G_{[0,6.7]}\mu^{h_1} \land F_{[41.7,48.4]}G_{[0,6.7]}\mu^{h_2} \land F_{[58.3,65.0]}G_{[0,6.7]}\mu^{h_3} \land F_{[83.33,90.0]}G_{[0,6.7]}\mu^{h_4}.$
The best resulting trajectory is shown in Fig. \ref{fig:iss inspection}. The algorithm was run for a total of 500 iterations, over 100 experiments with running time and trajectory cost summarized in Tab. \ref{tab:simulations summary}.

\vspace{-0.3cm}
\section{Conclusion}\label{sec:conclusion}
We introduced a sampling-based planning framework, based on RRT$^\star$, to synthesize trajectories under STL specifications with real-time performance. Namely, our approach leverages suitably constructed time-varying sets, with provable forward invariance guarantees with respect to controllable and input limited linear dynamics, to synthesize trajectories robustly satisfying a given STL task. As a next step, we aim to further expand our framework to nonlinear systems, applying techniques from spline optimization, and to broaden the class of STL specifications that we can consider.

\appendix
\subsection{Proof of Rule \ref{rule:eventually always rule}}\label{app:rule 4}
\begin{proof}
By \eqref{eq:gamma general}, select $\alpha \in a'\oplus [a,b]$ and $\beta = \alpha + (b'-a') $, from which  $\gamma^{\varphi}(t|\vec{\theta}^{\varphi}) = - r, \forall t \in [\alpha,\beta]$, such that $\mathcal{B}^{\varphi}(t|\vec{\theta}^{\varphi}) =\mathcal{H}^r, \; \forall t \in [\alpha,\beta]$. Thus if $\zeta_x \rhd_{[0,\beta]} \mathcal{B}^{\varphi}(t|\theta)$, then $h(\zeta_x(t)) \geq r,\, \forall t \in [\alpha,\beta] \Rightarrow \min_{[\alpha,\beta]}\{h(\zeta_x(t))\}\geq r$. By the robust semantics \eqref{eq:always robust}-\eqref{eq:eventually robust} we also get $\rho^\varphi(\zeta_x,0) = \max_{t \in [a,b]} \{ \min_{\tau \in t\oplus[a',b']}\{ h(\zeta_x(\tau) \}\}$. Let then $t^{\star} = \alpha - a' \in [a,b]$ such that, by the selection of $\alpha$ and $\beta$, we get $t^{\star}\oplus[a',b'] = [\alpha,\beta]$. Then  $\rho^\varphi(\zeta_x,0) = \max_{t \in [a,b]} \{ \min_{\tau \in t \oplus [a',b']}\{ h(\zeta_x(\tau)\}\} \geq  \min_{\tau \in t^{\star}\oplus [a',b']}\{ h(\zeta_x(\tau)\} = \min_{\tau \in [\alpha,\beta]}\{ h(\zeta_x(\tau)\} \geq r >0$, concluding that $(\zeta_x,0) \models_r \varphi$.
\end{proof}

\subsection{Proof of Proposition \ref{prop:always eventually decomposition}}\label{app:always eventually proof}

\begin{proof}
By the semantics \eqref{eq:always robust}, \eqref{eq:eventually robust} we have 
\begin{equation}\label{eq:some random passages I want to forget}
\rho^{\phi}(\zeta_x,0) = \min_{t \in [a,b]}  \{ \max_{\tau \in t\oplus [a', b']} \{h(\zeta_x(\tau)\}\}
\end{equation}
with robust satisfaction condition $\rho^{\phi}(\zeta_x,0) \geq r$. Take then any sequence of time instants $\tau_{w}$ with $\tau_{w} = \tau_{w-1} + \delta_{w} (b'-a')$ for some $\delta_w\in [0,1]$ and $\tau_0=a + a'$. We want to prove that if $h(\zeta_x(\tau_w))\geq r, \; \forall w\in[[n_{f}]]$, then $\rho^{\phi}(\zeta_x,0)\geq r$.  Start by noting that since $0 \leq \tau_{w}-\tau_{w-1} \leq (b'-a')$, then for any index $w$ it holds that $\tau_{w} \in t\oplus [a',b']$ for all $t\in -a'\oplus [\tau_{w-1},\tau_{w}]$. Hence, we can write 
\begin{equation}\label{eq:implication that we use}
\begin{aligned}
& h(\zeta_x(\tau_w))\geq r \Rightarrow \\
&\max_{\tau\in t\oplus [a',b']}\{  h(\zeta_x(\tau))\} \geq r,\; \forall t\in -a'\oplus [\tau_{w-1},\tau_{w}], \Rightarrow   \\
&\min_{t\in -a'\oplus [\tau_{w-1},\tau_{w}]}\{\max_{\tau\in t\oplus [a',b']}\{  h(\zeta_x(\tau))\} \} \geq r.
\end{aligned}
\end{equation}\
With the proposed iteration $\tau_{w} = \tau_{w-1} + \delta_w(b'-a'),\; \delta_w\in [0,1]$, consider now the union of intervals $\cup_{w\in [[n_f]]}-a'\oplus [\tau_{w-1},\tau_{w}] = -a'\oplus \cup_{w\in [[n_f]]} [\tau_{w-1},\tau_{w}]= -a' \oplus [a+a',\tau_{n_f}] =  [a,\tau_{n_f}-a']$.  By \eqref{eq:implication that we use} and recalling that for any union of sets $\cup_w[c_w,d_w]$ we have $
\min_{t\in \cup_w[c_w,d_w]}\{h(\zeta_x(t)\} \geq r \Leftrightarrow \min_{t\in [c_w,d_w]}\{h(\zeta_x(t)\}\geq r,\; \forall w$ then
$$
\begin{aligned}
&h(\zeta_x(\tau_{w}))\geq r, \forall w\in [[n_f]] \Rightarrow \text{by \eqref{eq:implication that we use}} \\
&\min_{t\in -a'\oplus [\tau_{w-1},\tau_{w}]}\{\max_{\tau\in t\oplus [a',b']}\{ h(\zeta_x(\tau))\} \} \geq r,\; \forall w\in [[n_f]] \Leftrightarrow \\
&\min_{t\in [a,\tau_{n_f}-a']}\{\max_{\tau\in t\oplus [a',b']}\{ h(\zeta_x(\tau))\} \} \geq r. 
\end{aligned}
$$
Compared to $\rho^{\phi}(\zeta_x,0)$ in \eqref{eq:some random passages I want to forget}, it is sufficient to have $\tau_{n_f} -a' \geq b$ in order to satisfy $\rho^{\phi}(\zeta_x,0) \geq r$. Since $\tau_{n_f} = a + a' + \sum_{w=1}^{n_f}\delta_w(b'-a')$ (by applying the recursion) then we must ensure 
$
\tau_{n_f}-a' \geq b \Rightarrow a + \sum_{w=1}^{n_f}\delta_w(b'-a') \geq b, \; \forall \sum_{w=1}^{n_f}\delta_w.
$
Let then the minimum constant coefficient $\underline{\delta} \in [0,1]$ such that the sum $a+\sum_{w=1}^{n_f} \underline{\delta} = a+ n_f \underline{\delta}  \geq b$. Then $n_f$ must be such that $n_f  \geq  \frac{1}{\underline{\delta}}\frac{b-a}{b'-a'}$. Since $\underline{\delta} \in [0,1]$, the minimum value for $n_f$ is given when $\underline{\delta} =1$ and the maximum is obtained when $\underline{\delta} = 0$ leading to $n_{f}=\infty$ (corresponding indeed to the satisfaction of the formula $G_{[a,b]}G_{[a',b']}\mu^h$). Since $n_f$ is an integer, then we conclude that the minimum lower bound on $n_f$ is given as $n_{f} \geq \lceil \frac{b-a}{b'-a'}\rceil$, with $\underline{\delta} = 1$. Furthermore, for a given selection of $n_{f}$ then $\underline{\delta} \geq  \frac{1}{n_f} \frac{b-a}{b'-a'}$ and the recursive formula for the selection of the time instants $\tau_w$ becomes $\tau_{w} = \tau_{w-1} + \delta_w(b'-a')$ with  $\delta_w \in [\underline{\delta},1] = [\frac{1}{n_f} \frac{b-a}{b'-a'},1]$. We thus arrived at the conclusion of the proof, since we just proved that for any given selection of $n_{f} \geq \lceil \frac{b-a}{b'-a'}\rceil$ if there exists a sequence of $\tau_{w}, w\in [[n_f]]$, respecting the given recursion with $h(\zeta_x(\tau_{w}))\geq r$, then $\rho^{\phi}(\zeta_x,0) \geq r$. Noting that $a_w = b_w = \tau_w$ in \eqref{eq:tau constraints} and that $(\zeta_x,0) \models_r F_{[a_w,b_w]} \mu^h \Leftrightarrow h(\zeta_x(\tau_w)) \geq r$, we have that, $(\zeta_x,0) \models_r \land_{w=1}^{n_f} F_{[a_w,b_w]} \Rightarrow h(\zeta_x(\tau_w)) \geq r, \; \forall w\in [[n_f]] \Rightarrow \rho^{\phi}(\zeta_x,0) = \min_{t\in [a,b]}\{\max_{\tau\in t\oplus [a',b']}\{ \rho^{\mu^h}(\zeta_x,\tau)\} \} \geq r \Rightarrow (\zeta_x,0)\models_{r} \phi$ concluding the proof.
\end{proof}
\subsection{Proof of Proposition \ref{prop:discrete viability prop}} \label{proof of discrete viability prop}
\begin{proof} Consider the sequence $(\beta_l)_{l=0}^{n_{\phi}}$ with $\beta_0=0$ with $\beta_{l} \geq \beta_{l-1}$. We have, by \eqref{eq:switch map}, that $\mathcal{I}^{\phi}(t) := \mathcal{I}^{\phi}_l \; \forall t \in [\beta_{l-1},\beta_{l})$ is constant in the half-open interval $[\beta_{l-1},\beta_{l})$. Moreover, by Prop \ref{prop:important properties}, we have that for every $\tilde{l} \in \mathcal{I}^{\phi}_l$ it holds $\mathcal{B}_{\tilde{l}}^{\varphi}(t) \supseteq \mathcal{B}_{\tilde{l}}^{\varphi}(\tau), \;  \forall \tau \in [t, \beta_l],\; \forall t\in [0,\beta_l] \supset (\beta_{l-1},\beta_l)$. Since $\mathcal{B}^{\phi}(t) = \cap_{\tilde{l}\in \mathcal{I}^{\phi}_l}\mathcal{B}_{\tilde{l}}^{\varphi}(t),\; \forall t\in  [\beta_{l-1},\beta_l)$ we conclude that $\mathcal{B}^{\phi}(t)\supseteq \mathcal{B}^{\phi}(\tau) ,\;\forall \tau \in [t, \beta_l),\; \forall t\in [\beta_{l-1},\beta_l)$. If there exists $\vec{\xi}_l\in \mathbb{X}$ such that $\vec{\xi}_l\in \lim_{\tau \rightarrow^{-} \beta_{l}}\mathcal{B}^{\phi}(\tau)$, then $\mathcal{B}^{\phi}(t) \supseteq \lim_{\tau \rightarrow^{-} \beta_{l}}\mathcal{B}^{\phi}(\tau) \supseteq\{\vec{\xi}_l\} \neq \emptyset $ for all $t\in [\beta_{l-1}, \beta_l)$. Since $\mathcal{B}^{\phi}(t)$ is viable, by Lemma \ref{lemma:conjunction}, then $\lim_{\tau \rightarrow^{-} \beta_{l}}\mathcal{B}^{\phi}(\tau) \subseteq \mathcal{B}^{\phi}(\beta_l)$ from which $\mathcal{B}^{\phi}(t)  \supseteq \{\vec{\xi}_l\} \neq \emptyset$ for all $t\in [\beta_{l-1}, \beta_l]$. The argument can now be repeated for every interval $(\beta_{l-1},\beta_{l})$ concluding that $\mathcal{B}^{\phi}(t) \neq \emptyset$ for all $t\in [\beta_{0}, \beta_{n_{\phi}}] = [0, \beta_{n_{\phi}}]$ if there exists  $\vec{\xi}_l\in \lim_{\tau \rightarrow^{-} \beta_{l}}\mathcal{B}^{\phi}(\tau) \; \forall l\in [[n_{\phi}]]$.
\end{proof}

\subsection{Proof of Proposition \ref{prop:preparatory convex hull representation}}\label{proof:preparatory convex hull representation}
\begin{proof}
    ($\Rightarrow$) The elements of $\mathcal{L}^i_{\mathfrak{b}^{\varphi}}(\vec{x},\vec{u},t)$ in \eqref{eq:convex hull lie derivative representation} are all scalars $\sum\limits_{k\in \mathcal{A}(\vec{x},t)} \lambda_k [\vec{d}_k^T\; e_{i}] (\bar{A}\vec{z} + \bar{B}\vec{u} + \bar{\vec{p}})$, for $\lambda_k \geq 0,\; \forall k\in \mathcal{A}(\vec{x},t)$, and $\sum_{k\in \mathcal{A}(\vec{x},t)} \lambda_k  =1 $. Recalling $\dot{\vec{z}} = \bar{A}\vec{z} + \bar{B}\vec{u} + \bar{\vec{p}}$, then for any element in $\mathcal{L}^i_{\mathfrak{b}^{\varphi}}(\vec{x},\vec{u},t)$ we have $\sum\limits_{k\in \mathcal{A}(\vec{x},t)} \lambda_k [\vec{d}_k^T\; e_{i}] \dot{\vec{z}} + \kappa(\mathfrak{b}(\vec{x},t)) = 
    \sum\limits_{k\in \mathcal{A}(\vec{x},t)} \lambda_k [\vec{d}_k^T\; e_{i}] \dot{\vec{z}} + \hspace{-0.4cm}\sum\limits_{k\in \mathcal{A}(\vec{x},t)} \hspace{-0.4cm}\lambda_k\kappa(\mathfrak{b}(\vec{x},t)) =
    \sum\limits_{k\in \mathcal{A}(\vec{x},t)} \lambda_k \left([\vec{d}_k^T\; e_{i}]\dot{\vec{z}}   + \kappa(\mathfrak{b}(\vec{x},t)) \right) \underset{ \eqref{eq:vertices condition}}{\geq} 0,$
    concluding that \eqref{eq:vertices condition} implies $\min \mathcal{L}^i_{\mathfrak{b}^{\varphi}}(\vec{x},\vec{u},t) \geq - \kappa(\mathfrak{b}(\vec{x},t))$. Note that the argument is independent of the index $i\in \{1,2\}$.
    ($\Leftarrow$) By contradiction, let $\min \mathcal{L}^i_{\mathfrak{b}^{\varphi}}(\vec{x},\vec{u},t) \geq - \kappa(\mathfrak{b}(\vec{x},t))$ be satisfied and  consider a single index $\tilde{k} \in \mathcal{A}(\vec{x},t)$ such that $[\vec{d}_{\tilde{k}}^T\; e_{i}] \dot{\vec{z}} <  - \kappa(\mathfrak{b}^{\varphi}(\vec{x},t))$. Let $\lambda_{\tilde{k}} = 1$ and $\lambda_{k} =0,\; \forall k\in \mathcal{A}(\vec{x},t)\setminus \{\tilde{k}\}$. Then $\sum_{k\in \mathcal{A}(\vec{x},t)} \lambda_k [d_k^T\; e_{i}]\dot{\vec{z}}  = [\vec{d}_{\tilde{k}}^T\; e_{i}] \dot{\vec{z}} \in \mathcal{L}^i_{\mathfrak{b}^{\varphi}}(\vec{x},\vec{u},t)$, which contradicts the initial assumption, as $\min \mathcal{L}^i_{\mathfrak{b}^{\varphi}}(\vec{x},\vec{u},t) \leq [\vec{d}^T_{\tilde{k}}\; e_{i}] \dot{\vec{z}} <  - \kappa(\mathfrak{b}^{\varphi}(\vec{x},t))$.
\end{proof}

\subsection{Proof of Proposition \ref{prop:min lie derivative prop}}\label{proof:min lie derivative prop}
\begin{proof}
By definition we have that $\mathfrak{b}^{\phi}(\vec{x},t) = \min_{l \in \mathcal{I}^{\phi}(t)} \{\mathfrak{b}_l^{\varphi}(\vec{x},t)\}$ in \eqref{eq:super min barrier} where $\mathfrak{b}_l^{\varphi}(\vec{x},t) = h_{l}(\vec{x}) + \gamma_l^{\varphi}(t)$. Since over the intervals $(s_j^{\phi},s_{j+1}^{\phi})$ the function $\mathfrak{b}_l^{\varphi}(\cdot,t)$ is differentiable in $t$ and $\mathfrak{b}_l^{\varphi}(\vec{x},\cdot)$ is the minimum of linear functions, then, the generalized gradient of $\partial^j \mathfrak{b}^{\phi}$ for each interval $(s^{\phi}_j, s^{\phi}_{j+1})$ is given by (see \cite[Prop. 7(iii)]{cortes2008discontinuous} and \cite[Example 15]{cortes2008discontinuous})
\begin{equation}\label{eq:some random mid passages I dislike}
\partial^j \mathfrak{b}^{\phi}(\vec{x},t) = \text{co}(\{ \partial^j \mathfrak{b}^{\varphi}_l(\vec{x},t) \mid \forall l\in \mathcal{A}^{\phi}(\vec{x},t) \cap \mathcal{I}^{\phi}(s_j) \}),
\end{equation}
where $\partial^j \mathfrak{b}^{\varphi}_l$ is the generalized gradient of each function $\mathfrak{b}^{\varphi}_l$ over the interval $(s^{\phi}_j, s^{\phi}_{j+1})$ and
\begin{equation}\label{eq:active set map for phi}
\mathcal{A}^{\phi}(\vec{x},t) = \{ l\in [[n_{\phi}]] \mid \mathfrak{b}_l^{\varphi}(\vec{x},t) = \mathfrak{b}^{\phi}(\vec{x},t)\},
\end{equation}
is the set of barriers equal to the minimum, while $\mathcal{I}^{\phi}(t)$ is the set of tasks that are still on-going at time $t$, and is constant over the interval $(s^{\phi}_j, s^{\phi}_{j+1})$, as per \eqref{eq:Constance relation}. Recalling that for any set of vectors $\{\vec{\nu}_j\}\subset \mathbb{R}^n$ and vector $\vec{c} \in \mathbb{R}^n$ it holds $\{\vec{\nu}^T \vec{c}\mid  \vec{\nu} \in co(\{ \vec{\nu}_j \})   \} = \text{co}(\{\vec{c}^T \vec{\nu}_j\})$, then we can replace the definition of $\partial^j \mathfrak{b}^{\phi}$ given in \eqref{eq:some random mid passages I dislike}, in the definition of weak set-valued Lie derivative in \eqref{eq:generalized lie derivative} to obtain the Lie derivative $\mathcal{L}^j_{\mathfrak{b}^{\phi}}(\vec{x},\vec{u},t)$ of $\mathfrak{b}^{\phi}$ over the interval $(s^{\phi}_j, s^{\phi}_{j+1})$ as
$$
\begin{aligned}
&\mathcal{L}^j_{\mathfrak{b}^{\phi}}(\vec{x},\vec{u},t) = \\
&\{\vec{\nu}^T (\bar{A}\vec{z} + \bar{B}\vec{u} + \bar{\vec{p}}) \mid \; \forall \vec{\nu}\in \text{co}(\{ \hspace{-0.7cm}\bigcup_{l\in \mathcal{A}^{\phi}(\vec{x},t) \cap \mathcal{I}^{\phi}(t)}\partial^j\mathfrak{b}^{\varphi}_l(\vec{x},t)\} )\} = \\
&\text{co} (\{ \vec{\nu}^T (\bar{A}\vec{z} + \bar{B}\vec{u}+\bar{\vec{p}})\mid \forall \vec{\nu}\in \hspace{-0.7cm}\bigcup_{l\in \mathcal{A}^{\phi}(\vec{x},t) \cap \mathcal{I}^{\phi}(t)}\partial^j\mathfrak{b}^{\varphi}_l(\vec{x},t) \}) =\\
&\text{co} ( \hspace{-0.5cm} \bigcup_{l\in \mathcal{A}^{\phi}(\vec{x},t) \cap \mathcal{I}^{\phi}(t)} \{ \vec{\nu}^T (\bar{A}\vec{z} + \bar{B}\vec{u}+\bar{\vec{p}})\mid \forall \vec{\nu}\in \partial^j\mathfrak{b}^{\varphi}_l(\vec{x},t) \}) =\\
&\text{co}(\hspace{-0.5cm}\bigcup_{ l\in \mathcal{A}^{\phi}(\vec{x},t) \cap \mathcal{I}^{\phi}(t)} \mathcal{L}^j_{\mathfrak{b}_l^{\varphi}}(\vec{x},\vec{u},t)),
\end{aligned}
$$
% \subseteq co(\bigcup_{ l\in \mathcal{I}^{\phi}(t)} \mathcal{L}^j_{\mathfrak{b}_l^{\varphi}}(\vec{x},\vec{u},t))
and by the definition of the convex hull we can write 
\begin{equation}\label{eq:subset condition}
\mathcal{L}^j_{\mathfrak{b}_l^{\varphi}}(\vec{x},\vec{u},t) \subseteq \mathcal{L}^j_{\mathfrak{b}^{\phi}}(\vec{x},\vec{u},t), \, \; \forall l\in \mathcal{A}^{\phi}(\vec{x},t) \cap \mathcal{I}^{\phi}(t).
\end{equation}
The following two preparatory propositions are then needed.
\begin{proposition}\label{prop:prop1 support}
For any $(\vec{x},t,\vec{u}) \in \mathbb{R}^n\times (s_j^{\phi},s_{j+1}^{\phi}) \times \mathbb{U}$ if $\min \mathcal{L}^j_{\mathfrak{b}_l^{\varphi}}(\vec{x},\vec{u},t) \geq -\kappa(\mathfrak{b}^{\phi}(\vec{x},t) ), \;\forall l\in \mathcal{A}^{\phi}(\vec{x},t) \cap \mathcal{I}^{\phi}(t)$ then $\min \mathcal{L}^j_{\mathfrak{b}^{\phi}}(\vec{x},\vec{u},t) \geq -\kappa(\mathfrak{b}^{\phi}(\vec{x},t))$.
\end{proposition}
\begin{proof}
    By contradiction, consider an index $\tilde{l}\in \mathcal{A}^{\phi}(\vec{x},t) \cap \mathcal{I}^{\phi}(t)$ such that  $\min \mathcal{L}^j_{\mathfrak{b}_l^{\varphi}}(\vec{x},\vec{u},t) < -\kappa(\mathfrak{b}^{\phi}(\vec{x},t) )$, while $\min \mathcal{L}^j_{\mathfrak{b}^{\phi}}(\vec{x},\vec{u},t) \geq -\kappa(\mathfrak{b}^{\phi}(\vec{x},t))$. Since $\mathcal{L}^j_{\mathfrak{b}_l^{\varphi}}(\vec{x},\vec{u},t)\subseteq \mathcal{L}^j_{\mathfrak{b}^{\phi}}(\vec{x},\vec{u},t)$, as per \eqref{eq:subset condition}, then $\min \mathcal{L}^j_{\mathfrak{b}^{\phi}}(\vec{x},\vec{u},t) \leq \min \mathcal{L}^j_{\mathfrak{b}_l^{\varphi}}(\vec{x},\vec{u},t) < -\kappa(\mathfrak{b}^{\phi}(\vec{x},t) )$ which is a contradiction, concluding the proof.
\end{proof}

\begin{proposition}\label{prop:prop2 support}
For any $(\vec{x},t,\vec{u}) \in \mathbb{R}^n\times (s_j^{\phi},s_{j+1}^{\phi}) \times \mathbb{U}$, if $\min \mathcal{L}^j_{\mathfrak{b}_l^{\varphi}}(\vec{x},\vec{u},t) \geq -\kappa(\mathfrak{b}_l^{\varphi}(\vec{x},t) ), \;\forall l\in \mathcal{I}^{\phi}(t)$ then $\min \mathcal{L}^j_{\mathfrak{b}^{\phi}}(\vec{x},\vec{u},t) \geq -\kappa(\mathfrak{b}^{\phi}(\vec{x},t))$.
\end{proposition}
\begin{proof}
    By the definition of $\mathcal{A}^{\phi}(\vec{x},t)$ in \eqref{eq:active set map for phi} we have that for all $l\in \mathcal{A}^{\phi}(\vec{x},t) \cap \mathcal{I}^{\phi}(t)$ it holds $\mathfrak{b}_l^{\varphi}(\vec{x},t) = \mathfrak{b}^{\phi}(\vec{x},t)$ while $\mathfrak{b}_l^{\varphi}(\vec{x},t) > \mathfrak{b}^{\phi}(\vec{x},t)$ for all $l\in  \mathcal{I}^{\phi}(t)\setminus \mathcal{A}^{\phi}(\vec{x},t)$. We thus have that $\min \mathcal{L}^j_{\mathfrak{b}_l^{\varphi}}(\vec{x},\vec{u},t) \geq -\kappa(\mathfrak{b}_l^{\varphi}(\vec{x},t) ), \;\forall l\in \mathcal{I}^{\phi}(t)$ implies $\min \mathcal{L}^j_{\mathfrak{b}_l^{\varphi}}(\vec{x},\vec{u},t) \geq -\kappa(\mathfrak{b}_l^{\varphi}(\vec{x},t) ) = - \kappa(\mathfrak{b}_{\phi}(\vec{x},t) ), \;\forall l\in \mathcal{A}^{\phi}(\vec{x},t) \cap \mathcal{I}^{\phi}(t)$ which implies $\min \mathcal{L}^j_{\mathfrak{b}^{\phi}}(\vec{x},\vec{u},t) \geq -\kappa(\mathfrak{b}^{\phi}(\vec{x},t))$ by Prop. \ref{prop:prop1 support}
\end{proof}

Hence, by Proposition \ref{prop:prop2 support} we derived the implication $\min \mathcal{L}^j_{\mathfrak{b}_l^{\varphi}}(\vec{x},\vec{u},t) \geq -\kappa(\mathfrak{b}_l^{\varphi}(\vec{x},t) ), \;\forall l\in \mathcal{I}^{\phi}(t) \\
 \Rightarrow \min \mathcal{L}^j_{\mathfrak{b}^{\phi}}(\vec{x},\vec{u},t) \geq -\kappa(\mathfrak{b}^{\phi}(\vec{x},t))$. As a last step, note that for each $l\in \mathcal{I}^{\phi}(t)$ either $(s_1^l,s_2^l) \supseteq (s_j^{\phi},s_{j+1}^{\phi})$ or $(s_2^l,s_3^l) \supseteq (s_j^{\phi},s_{j+1}^{\phi})$, meaning $(s_j^{\phi},s_{j+1}^{\phi})$ is contained in one of the two linear sections of each function $\gamma^{\varphi}_l$, such that $\mathcal{L}^j_{\mathfrak{b}_l^{\varphi}}(\vec{x},\vec{u},t)$ over the interval $(s_j^{\phi},s_{j+1}^{\phi})$ is equal to \eqref{eq:convex hull lie derivative representation} as
$$
\mathcal{L}^j_{\mathfrak{b}_l^{\varphi}}(\vec{x},\vec{u},t) = \mathcal{L}^i_{\mathfrak{b}_l^{\varphi}}(\vec{x},\vec{u},t), \; i=\Upsilon_l(t), \forall t\in (s_j^{\phi},s_{j+1}^{\phi}).
$$
where the map $\Upsilon_l(t)$ is constant over $(s_j^{\phi},s_{j+1}^{\phi})$ and identifies which linear sections of $\gamma^{\varphi}_l$ the interval $(s_j^{\phi},s_{j+1}^{\phi})$ is spanning. Thus we showed that for each interval $(s_j^{\phi},s_{j+1}^{\phi})$ it holds 
$
\min \mathcal{L}^i_{\mathfrak{b}_l^{\varphi}}(\vec{x},\vec{u},t) \geq -\kappa(\mathfrak{b}_l^{\varphi}(\vec{x},t) ), i = \Upsilon(t), \;\forall l\in \mathcal{I}^{\phi}(t) 
 \Rightarrow \min \mathcal{L}^j_{\mathfrak{b}^{\phi}}(\vec{x},\vec{u},t) \geq -\kappa(\mathfrak{b}^{\phi}(\vec{x},t))
$ as we wanted to prove.
\end{proof}

%\section*{References}
\vspace{-0.3cm}
\bibliographystyle{ieeetr} 
\bibliography{references}

\end{document}